\documentclass[12pt]{article}
\usepackage{amsmath}
\usepackage{amsfonts,amssymb}
\usepackage{mathrsfs} 
\usepackage{graphicx,psfrag,epsf}
\usepackage{enumerate}
\usepackage{xcolor}
\usepackage{enumitem}
\usepackage{color}
\usepackage{multirow}
\usepackage{amsthm}
\usepackage{subcaption}
\usepackage{hyperref}

\hypersetup{
    colorlinks,
    linkcolor={red!50!black},
    citecolor={blue!50!black},
    urlcolor={blue!80!black}
}

\usepackage[authoryear]{natbib}
\usepackage{verbatim}

\newtheorem{theorem}{Theorem}
\newtheorem{corollary}{Corollary}
\newtheorem{proposition}{Proposition}
\newtheorem{example}{Example}
\newtheorem{claim}{Claim}
\newtheorem{definition}{Definition}

\newtheoremstyle{rem}
{\topsep}
{\topsep}
{}
{0pt}
{\itshape}
{.}
{ }
{\thmname{#1}\thmnumber{ #2}\thmnote{ (#3)}}

\theoremstyle{rem}
\newtheorem{remark}{Remark}

\usepackage[width=6.25in, height=8.5in]{geometry}

\allowdisplaybreaks

\newcommand{\lin}[1]{#1}
\newcommand{\mylabel}[2]{#2\def\@currentlabel{#2}\label{#1}}

\usepackage{setspace}
\global\long\def\tp{\top}
\global\long\def\tr{\mathrm{tr}}
\global\long\def\ltwo{L^{2}}
\global\long\def\matvec#1{\mathbf{#1}}
\global\long\def\bs{\phi}
\global\long\def\absconst{\theta}
\global\long\def\Op{O_{P}}
\global\long\def\vec#1{\mathbf{#1}}
\global\long\def\expect{\mathbb{E}}
\global\long\def\diffop{\mathrm{d}}
\global\long\def\prob{\mathbb{P}}
\global\long\def\var{\mathbb{V}}
\def\ltwo{L^2}
\def\diffop{\mathrm{d}}
\def\tdomain{\mathcal{T}}
\def\expect{\mathbb{E}}

\def\real{\mathbb{R}}
\def\fronorm{\mathrm{F}}
\def\hsnorm{\ltwo}
\def\DD{\eta}

\begin{document}

  \title{\bf Basis Expansions for Functional Snippets} 
  \author{Zhenhua Lin$^{\dag,1}$ \and Jane-Ling Wang$^{\ddag,2}$ \and Qixian Zhong$^{\S,3}$}
  \date{%
    $^\dag$\textit{National University of Singapore} \quad \quad \quad \\
    $^{\ddag}$\textit{University of California, Daivs}\\
    $^\S$\textit{Tsinghua University}
}
  \maketitle
  
  \footnotetext[1]{Email: stalz@nus.edu.sg. Research was supported by NIH ECHO grant  5UG3OD023313-03 and 4UH3-OD023313-04, and NUS start-up grant R-155-001-217-133.}
  \footnotetext[2]{Email: janelwang@ucdavis.edu. Research was supported by NSF grants 15-12975 and 19-14917,  and NIH ECHO grant 4UH3-OD023313-04.}
   \footnotetext[3]{Email: zqx15@mails.tsinghua.edu.cn. Research was supported by National Science Foundation of China (NSFC11771241) and Chinese Government Scholarship (CSC201806210163).}

\begin{abstract}
Estimation of mean and covariance functions is fundamental for functional data analysis. While this topic has been studied extensively in the literature, a key assumption is that  there are enough data in the domain of interest to estimate both the mean and covariance functions. In this paper, we investigate mean and covariance estimation for functional snippets in which observations from a subject are available only in an interval of length strictly (and often much) shorter than the length of the whole interval of interest. For such a sampling plan,  no data is available for direct estimation of the off-diagonal region of the covariance function. We tackle this challenge via a basis representation of the covariance function. The proposed approach allows one to consistently estimate an infinite-rank covariance function from functional snippets. We establish the convergence rates for the proposed estimators and illustrate their finite-sample performance via simulation studies and two data applications.

\vspace{3mm}
  \end{abstract}

\section{Introduction}
\label{sec:intro}
Nowadays functional data are commonly encountered in practice, due to the advances in modern science and technology that enhance capability of data collection, storage and processing. Both unsupervised learning approaches, such as dimension reduction via functional principal component analysis \citep{Rao1958,Dauxois1982,Hall2009,Mas2015},  and supervised learning, such as  functional regression  \citep{Cardot1999,Mueller2005,Ferraty2006,Hall2007c,Mueller2008,Kong2016} are well studied in the literature. For a comprehensive treatment of these subjects, we recommend the monographs \cite{Ramsay2005}, \cite{Ferraty2006}, \cite{Horvath2012},  \cite{Hsing2015} and \cite{Kokoszka2017}, and the review papers \cite{WangCM16} and \cite{Aneiros:2019aa}.

Critical to the statistical analysis of such data is 
the estimation of the mean and covariance functions, since they are the foundation of the aforementioned unsupervised and supervised learning tasks. {For example, covariance estimation is a critical step to functional principal component analysis, as illustrated in Section \ref{sec:application}.  In addition,  covariance estimation is often required for functional regression or classification for functional data.}
In reality, functions can only be recorded at a set of discrete points on the domain of the functions, where this set  may vary among subjects and the measurements may contain noise. Estimation of  mean and covariance functions in this context has been extensively studied by 
\cite{Rice1991}, \cite{Cardot2000nonparametric},  \cite{James2000}, \cite{Cai2010nonparametric}, \cite{Cai2011},  
\cite{Yao2005b}, \cite{Li2010} and \cite{Zhang2016}, among many others. 
In addition to the discrete nature of the observed functional data,  
subjects often  stay in the study only for a subject-specific period that
 is much shorter than the span of the whole study.  This usually does not cause much of a problem for mean estimation but brings challenges to  covariance estimation. 

For illustration and without loss of generality, we assume that the domain of the functional data $X(t)$ is the unit interval $\tdomain=[0,1]$ and each subject only stays in the study for a period of length $\delta<1$. Data with these  characteristics are termed ``functional snippets'',  in analogy to the longitudinal snippets analyzed in \cite{Dawson2018}. For such data, there is no  information in the off-diagonal region $\tdomain_{\delta}^c:=\{(s,t)\in[0,1]^2:|s-t|>\delta\}$ of the covariance function Cov$(X(s), X(t))$, and therefore there is no local information available for estimating the covariance function in this region. Mathematically, this amounts to 
\begin{equation}\label{eq:fs}\prob\{(\cup_{i=1}^n[A_i,B_i]^2)\cap \tdomain_{\delta}^c=\emptyset\}=1, \end{equation} for some $\delta>0$ and for all $n$, where $[A_i,B_i]$ denotes the subinterval on which $X_i$ is observed.  Figure \ref{fig:cov}(a) illustrates such a situation for the bone mineral density data, where the band with available data is prominent and narrow. Another example  for a wider band is illustrated in Figure \ref{fig:sbp-cov}(a) for the systolic blood pressure data. Both data will be further  studied in Section \ref{sec:application}.  For this type of data, local smoothing methods, e.g.,  the aforementioned \cite{Yao2005a}  and \cite{Li2010}, fail to yield a consistent estimate of the covariance function in the off-diagonal region as they are primarily interpolation methods.

The literature of statistical analysis for  functional snippets is in its infancy.  \cite{Delaigle2016} proposed to approximate snippets by segments of Markov chains. This  method is only valid at the discrete level, as shown by \cite{Descary2019recovering}. More recently, to analyze functional snippets,  \cite{Zhang2018} and \cite{Descary2019recovering} used matrix completion techniques that innately work on a common grid of the domain $\tdomain$. These approaches require modification when  snippets are recorded on random and irregular design points, which is often encountered in  applications.  Yet,  published theoretical analyses focus on the regular and dense design. 

{Another challenge is that estimating the covariance function for functional snippets is an extrapolation problem, which requires additional  assumptions to overcome. The minimal identifiability condition for functional snippets is the trivial condition that the covariance function in the observable band $\tdomain_\delta$ determines the covariance function on the entire domain. This minimal identifiability is a high-level concept \citep{Delaigle2019}, and  consequently, existing works attempt to find some specialized conditions that imply the above minimal identifiability assumption. For instance,  \cite{Descary2019recovering} assumed that the covariance function is an analytic function, and \cite{Delaigle2019} proposed the linear predictability assumption, which assumes that the values of the process $X(t)$ on a subinterval  can be linearly predicted by the values of the same process on  another subinterval.}  

{In contrast to the aforementioned approaches, that impose a particular assumption on the process $X(t)$ itself or on its covariance function, we define  identifiability through  a  family $\mathcal{C}$ in which the  covariance function resides, and term such a family  $\mathcal T_{\delta}$-identifiable if any two members from the family $\mathcal{C}$ and identical on the diagonal region $\mathcal T_{\delta}$ are equal everywhere. The $\tdomain_{\delta}$-identifiability is the same as the above minimal identifiability except that we make the reference to a family explicit. The family $\mathcal{C}$ is comparable to the traditional parameter space or model, so our definition of identifiability is in line with the  conventional statistical concept of identifiability that are imposed on the model. This concept of identifiability is rather general and encompasses the aforementioned identifiability assumptions as special cases. For example, the class of analytic functions considered in \cite{Descary2019recovering} and the class of covariance functions associated with linearly predictable random processes are $\tdomain_{\delta}$-identifiable families; see Examples \ref{ex:af} and \ref{ex:lp} for details. The primary reason that we adopt this minimal identifiability is that, our method and theory to be developed in Sections \ref{sec:meth} and \ref{sec:asyth} apply to all $\tdomain_{\delta}$-identifiable families under some regularity conditions.} 

Like functional snippets, fragments are also partially observed functional data and have been explored by many, such as
\cite{liebl2013modeling}, \cite{Gellar:2014aa}, \cite{Goldberg:2014aa}, \cite{Kraus2015}, \cite{Gromenko:2017aa}, \cite{Stefanucci2018},  \cite{Mojirsheibani2018}, \cite{Kraus2019},  \cite{Kneip:2019} and \cite{Liebl:2019aa}.  However, the problem to recover the covariance function for functional fragments is often formulated as an interpolation problem,  e.g., the work of \cite{Kneip:2019} assumes that    $\prob([A_i,B_i]^2=[0,1]^2)>0$, which implies  information and design points for the off-diagonal region $\tdomain_\delta^c$ are still available. In contrast,  information for that region is completely missing in the context of functional snippets \eqref{eq:fs}, which significantly elevates the difficulty of statistical analysis.   For this reason, we adopt the term functional snippets to distinguish them from fragments or other partially observed  functional data.

Under the umbrella of $\tdomain_\delta$-identifiability,    
we propose to approach functional snippets from the perspective of basis expansion. The main idea is to represent the covariance function by basis functions composed from tensor products of  analytic orthonormal functions defined on $\tdomain$. Basis functions, in particular spline basis functions, have been extensively explored in both nonparametric smoothing and functional data analysis by \cite{Wahb1990}, \cite{Wood2003},  \cite{Rice2001}, \cite{Ramsay2005} and \cite{Crambes2009}, among many others. However, they are not suited for the extrapolation problem of functional snippets, as these bases are local. Unlike spline bases that are controlled by knots, analytic bases are global, in the sense that they are  independent of local information such as knots or design points and are completely determined by their values on a countably infinite subset of the interval $\tdomain$. This feature of analytic bases allows information to pass from the diagonal region to the off-diagonal region along the basis functions. Consequently,  the missing pieces of the covariance function can then be inferred from the data available in the diagonal region when the covariance function is from a  $\tdomain_\delta$-identifiable class.  In contrast, this is generally impossible for B-spline or other local bases.

{In addition to the minimal identifiability assumption, the consistency of the proposed estimation method requires regularity conditions to overcome the challenges of extrapolation. One regularity condition that we identified is the \emph{bounded sequential compactness} (BSC) of the family $\mathcal C$ of covariance functions under consideration. This new concept, developed in Section \ref{subsec:cov-theory}, essentially controls the complexity (or size) of the family $\mathcal C$ and enables us to establish the consistency and convergence rate of the proposed estimator in a nonparametric extrapolation setting. This condition is mild; for example, all families of functions that are uniformly bounded and Lipschitz continuous with a common Lipschitz constant are BSC families, as shown in Section \ref{subsec:cov-theory}. Such a regularity condition, not seen in the literature, is not required for interpolation and thus intrinsically   separates nonparametric extrapolation from interpolation.
} 

{When our work was completed, we became aware of a related work that is independently developed by \cite{Delaigle2019}. Although the work also uses a basis expansion approach, it is substantially different from ours. First, it focuses more on development of identifiability conditions while ours is on methodological development and theoretical analysis of the proposed estimators. Second, the method of \cite{Delaigle2019} extrapolates a pilot estimate from the diagonal region to the entire region by basis expansion without regularization. This may lead to excessive variability of the estimator in the off-diagonal region. In contrast, our method estimates the basis coefficients directly from data with penalized least squares and does not require a pilot estimate. Third, the convergence rate established in that work hinges on and thus is limited by the convergence rate of the pilot estimate, while our analysis gives an explicit rate that is adaptive to the smoothness of the underlying covariance function. Finally, the work of \cite{Delaigle2019} heuristically includes the identity function $f(x)=x$ into the Fourier basis to handle nonperiodicity, while ours adopts the Fourier extension technique that is well established in the field of numerical analysis.}

{In summary, this paper makes the following contributions to the field of functional snippet analysis. First, we develop a computationally efficient estimator for the covariance function based on basis expansion under the minimal identifiability condition. Second, we propose the novel concept of bounded sequential compactness, based on which we establish an explicit and adaptive rate of convergence for the proposed estimator without relying on any specific assumptions for identifibility. Such concept is new in the literature of nonparametric regression. Third, we introduce the well established numerical techniques, Fourier extension and geometric Newton method, to the field of functional data analysis for the first time.}

\section{Methodology}
\label{sec:meth}
Let $\{X(t):t\in \tdomain\}$ be a second-order stochastic process on a compact interval $\tdomain\subset\mathbb R$, which without loss of generality is taken to be $[0,1]$. The mean and covariance functions of $X$ are defined as $\mu_{0}(t)=\expect X(t)$ and $\gamma_{0}(s,t)=\text{Cov}(X(s),X(t)),$ respectively. The observed functions $X_1,\ldots,X_n$ are statistically modeled as independent and identically distributed realizations of $X$. In practice, each realization $X_i$ is only recorded at subject-specific $m_i$ time points $T_{i1},\ldots, T_{im_i}$ with measurement errors. More precisely, for functional snippets, the observed data  are pairs $(T_{ij},Y_{ij})$, where
\begin{equation}
	\begin{aligned}
	Y_{ij}=X_i(T_{ij})+\varepsilon_{ij},\quad i=1,\ldots,n, \quad j=1,\ldots,m_i,
	\end{aligned}
\end{equation}
 $\varepsilon_{ij}$ is the random noise with mean zero and unknown variance $\sigma^2$, and there is a $\delta\in(0,1)$ for which $|T_{ij}-T_{ik}|\le \delta $ for all $i$, $j$ and $k$. 
The focus of this paper is to estimate the mean and covariance functions of $X$ using these data pairs $(T_{ij},Y_{ij})$.

\subsection{Mean Function}\label{subsec:mean}

Although functional snippets pose a challenge for covariance estimation, they usually do not obstruct  mean estimation, since data for the estimation are available likely across the whole domain of interest. In light of this observation, traditional methods such as local linear smoothing \citep{Yao2005b, Li2010} can be employed. Below we adopt an analytic basis expansion approach. The advantage of this approach is its computational efficiency and adaptivity to the  regularity of the underlying mean function $\mu_0$; see also Section \ref{subsec:mean-theory}.

Let $\boldsymbol{\Phi}=\{\phi_1,\ldots\}$ be a complete orthonormal basis of $\ltwo(0,1)$ that consists of squared integrable functions defined on the interval $[0,1]$. When $\mu_0\in\ltwo(0,1)$, it can be represented by the following series in terms of the basis $\boldsymbol\Phi$,
\begin{equation*}
\mu_0(t)=\sum_{k=1}^\infty a_k\phi_k(t),
\end{equation*}
where $a_k=\int_0^1 \mu_0(t)\phi_k(t)\diffop t$. In practice, one often approximates such series by its first $q>0$ leading terms, where $q$ is a tuning parameter controlling the approximation quality. The coefficients $a_1,\ldots,a_q$ are then estimated from data by penalized least squares. Specifically, with the notation $\boldsymbol{\Phi}_q(t)=(\phi_{1}(t), \cdots, \phi_{q}(t))^{\tp}\in\mathbb{R}^{q}$ and $\mathbf a_0=(a_1,\ldots,a_q)^{\tp}$, the estimator of $\mathbf a_0=(a_1,\ldots,a_q)^{\tp}$ is given by 

\begin{equation}\label{eq:mean-estimator}
	\begin{aligned}
	\hat{\mathbf a}=\mathop{\arg\min}\limits_{\textbf{a}\in \mathbb{R}^{q}}\Big\{\sum_{i=1}^{n}v_{i}\sum_{j=1}^{m_{i}}[Y_{ij}-\textbf{a}^{\tp}\boldsymbol{\Phi}_q(T_{ij})]^{2}+\rho H(\textbf{a}^{\tp}\boldsymbol{\Phi}_q)  \Big\},
	\end{aligned}
\end{equation}
and $\mu_0$ is estimated by $$\hat\mu(t)=\hat {\vec a}^\top \boldsymbol\Phi_q(t),$$
where the weights $v_{i}>0$ satisfy $\sum_{i=1}^{n}v_{i}m_{i}=1$, $H(\cdot)$ represents the roughness penalty, and $\rho$ is a tuning parameter that provides trade-off between the fidelity to the data and the smoothness of the estimate. 
There are two commonly used schemes for the weights, equal weight per observation (OBS) and equal weight per subject (SUBJ), for which the weights $v_{i}$ are $1/(\sum_{i=1}^{n}m_{i})$ and $1/(nm_{i})$, respectively.  These and also alternative weight schemes are discussed in  \cite{Zhang2016, Zhang2018a}. 

The penalty term in \eqref{eq:mean-estimator} is introduced to prevent excessive variability of the estimator when a large number of basis functions are required to adequately approximate $\mu_0$ and when the sample size is not sufficiently large. In the asymptotic analysis of $\hat\mu$ in Section \ref{sec:asyth}, we will see that this penalty term does not affect the convergence rate of $\hat\mu$ when the tuning parameter $\rho$ is not too large. In our study, the roughness penalty is given by
\begin{equation*}
H(g)=\int_{0}^{1}\{g^{(2)}(t)\}^{2}\diffop t,
\end{equation*}
where $g^{(2)}$ denotes the second derivative of $g$. In matrix form, for $g(t)=\mathbf a^\tp\boldsymbol{\Phi}_q(t)$, it equals $\textbf{a}^{\tp}\textbf{W}\textbf{a}$, where  $\mathbf W$ is a $q\times q$ matrix with elements $W_{kl}=\int_{0}^{1}\phi_k^{(2)}(t)\phi_l^{(2)}(t)\diffop t$. The choices of $q$ and $\rho$ are discussed in Section \ref{subsec:computation}.

\subsection{Covariance Function}\label{subsec:cov}
Since functional snippets do not provide any direct information for the off-diagonal region, the only way to recover the covariance in the off-diagonal region  is to infer it from the diagonal region.  
{The following definition formulates this basic requirement for identifiability.} 
\begin{definition}
    {A family $\mathcal C$ of covariance functions is called a $\tdomain_\delta$-identifiable family if $\gamma_1,\gamma_2\in\mathcal C$ and $\gamma_1(s,t)=\gamma_2(s,t)$ for all $(s,t)\in\tdomain_\delta$ imply that $\gamma_1(s,t)=\gamma_2(s,t)$ for all $(s,t)\in\tdomain^2$.}
\end{definition}

 \lin{Intuitively, we consider a family $\mathcal C$ of covariance functions and require the covariance functions to be uniquely identified \emph{within} the family $\mathcal C$  by their values on the diagonal region. Below we provide four examples to illustrate the ubiquitousness of  $\tdomain_\delta$-identifiable families.}
\begin{example}[Analytic functions]\label{ex:af}
A function is analytic if it can be locally represented by a convergent power series. According to Corollary 1.2.7 of \cite{krantz2002primer}, if two analytic functions agree on $\tdomain_\delta$, then they are identical on $[0,1]^2$. Thus, the family of analytic functions is a $\tdomain_\delta$-identifiable family as observed by \cite{Descary2019recovering}, who also provided an elegant counterexample to demonstrate that $C^{\infty}([0,1]^2)$, the space of infinitely differentiable functions,  is not $\tdomain_\delta$-identifiable. 
\end{example}
\begin{example}[Sobolev sandwich families]\label{ex:ssf}For any $0<\epsilon<\delta$, consider the family of continuous functions that belong to a two-dimensional Sobolev space on $\tdomain_\epsilon$ and are analytic elsewhere. Such functions have an $r$-times differentiable diagonal component sandwiched between two analytic off-diagonal pieces. The family is $\tdomain_\delta$-identifiable, because the values of such functions on the off-diagonal region are fully determined by the values on the uncountable set $\tdomain_\epsilon^c\cap\tdomain_\delta\subset\tdomain_\delta$ according to Corollary 1.2.7 of \cite{krantz2002primer}. Note that this family contains   functions with derivatives only up to a finite order.
\end{example}
\begin{example}[Semiparametric families]\label{ex:spf}
Consider the family of functions of the form $g(s)h(s,t)g(t)$, where $g$ is a function from a nonparametric class $\mathcal G$ and $h$ is from a parametric class  $\mathcal H$ of correlation functions.  This  family, considered in \cite{Lin2019}, is generally $\tdomain_\delta$-identifiable,  as long as both $\mathcal G$ and $\mathcal{H}$ are identifiable. For instance, when $\mathcal G$ is a one-dimensional Sobolev space and $\mathcal H$ is the class of Mat\'ern correlation functions, the family is  $\tdomain_\delta$-identifiable. Note that no analyticity is assumed for  this family.
\end{example}

\begin{example}[Linearly predictable families]\label{ex:lp}
	\lin{For a random process $X$ defined on $\tdomain$, we say that the snippet $\{X(t)\}_{t\in I^\ast}$ on the subinterval $I^\ast\subset \tdomain$ is linearly $(B,\epsilon)$-predictable \citep{Delaigle2019} from another subinterval $I\subset\tdomain$, if for all $t\in I^\ast$, there exists an integrable function $L_t(s)$ defined on $I$  such that $\sup_{t\in I^\ast}\sup_{s\in I}|L_t(s)|<\infty$, $\sup_{t\in I^\ast}\int_I|L_t(s)|\diffop s<{B}$ and $X(t)=\mu(t) + \int_I L_t(s)\{X(s)-\mu(s)\}\diffop s+Z(t)$, where for all $t\in I^\ast$, $Z(t)$ is a zero mean random variable such that $\expect Z^2(t)\leq \epsilon^2$. 
	Fix an integer $h>0$ and a partition $I_0,\ldots,I_h$ of $\tdomain$ such that $I_j\times I_j\subset \tdomain_{\delta}$. Consider only the class $\mathscr X$ of random processes $X$ whose snippet $\{X(t)\}_{t\in I_j}$ is  linearly $(B_j,\epsilon_j)$-predictable from $I_{j^\ast}$ for some $0\leq j^\ast\leq j-1$ and all $\epsilon_j>0$, and for all $j=1,\ldots,h$. Let $\mathcal L$ be the collection of covariance functions of random processes in $\mathscr X$. By  \cite{Delaigle2019}, each member in $\mathcal L$ is identifiable within $\mathcal L$ from its values on the diagonal region $\tdomain_{\delta}$, and thus $\mathcal L$ is $\tdomain_\delta$-identifiable.}
\end{example}

 With the  $\tdomain_\delta$-identifiability of the family $\mathcal{C}$, it is now possible to infer the off-diagonal region by the information contained in the raw covariance  $\Gamma_{ijk}=\{Y_{ij}-\hat{\mu}(T_{ij})\}\{Y_{ik}-\hat{\mu}(T_{ik})\}$ available only in the diagonal region. 
To this end,  
we propose to transport  information from the diagonal region to the off-diagonal region through the basis functions $\phi_k\otimes\phi_l$ with  $(\phi_k\otimes\phi_l)(s,t)=\phi_k(s)\phi_l(t)$ for $s,t\in\tdomain$, by approximating $\gamma_0$ with 
\begin{equation}\label{fsanalytic}
\gamma_{{\vec C}_0}(s,t)= \sum_{1 \leq k,l \leq p} c_{kl}\phi_{k}(s)\varphi_{l}(t),\quad (s,t) \in [0,1]^2,
\end{equation}
where $ c_{kl}=\iint \gamma_0(s,t)\phi_k(s)\phi_l(t)\diffop s\diffop t$, ${\vec C}_0$ is the matrix of coefficients $c_{kl}$, and $p\geq 1$ is an integer. There are countless bases that can serve in \eqref{fsanalytic}; however, if we choose an analytic basis $\boldsymbol{\Phi}$, then their values in the diagonal region completely determine their values in the off-diagonal region. When such a representation  of the covariance function $\gamma_0$ is adopted and the unknown coefficients $c_{kl}$ are estimated from data, the information  contained in the estimated coefficients extends  from the diagonal region to the off-diagonal region  through the analyticity of the basis functions.  

To estimate the coefficients $c_{kl}$ from data, we  adopt the idea of penalized least squares, where the squared loss of a given function $\gamma$ is measured by the sum of weighted squared errors   
$$SSE(\gamma)=\sum_{i=1}^n w_i \sum_{1\leq j\neq k\leq m_{i}}\{\Gamma_{ijl}-\gamma(T_{ij},T_{il})\}^2,$$ 
where $w_{i}>0$ are weights satisfying $\sum_{i=1}^{n}m_{i}(m_{i}-1)w_{i}=1$, while the roughness penalty is given by $$J(\gamma)=\iint\frac{1}{2}\left\{\left[\frac{\partial^2\gamma}{\partial s^2}\right]^2+2\left[\frac{\partial^2\gamma}{\partial s\partial t}\right]^2+\left[\frac{\partial^2\gamma}{\partial t^2}\right]^2\right\}\diffop s\diffop t.$$ Then,   \begin{equation}
{J}(\gamma_{\mathbf C})=\tr(\textbf{C}\textbf{U}\textbf{C}\textbf{W})+\tr(\textbf{C}\textbf{V}\textbf{C}\textbf{V}),
 \end{equation}
where $\gamma_{\mathbf C}$ is defined in \eqref{fsanalytic} with ${\vec C}_0$ replaced by $\vec C$, $\tr$ denotes the matrix trace, and $\mathbf U,\mathbf V$ are $p\times p$ matrices with elements $U_{kl}=\int_{0}^{1}\phi_k(s)\phi_l(s)\diffop s$ and $V_{kl}=\int_{0}^{1}\phi_k^{(1)}(s)\phi_l^{(1)}(s)\diffop s$, respectively. The estimator $\hat{\gamma}(s,t)$ of $\gamma_0(s,t)$ is then taken as $\hat\gamma(s,t)=\boldsymbol{\Phi}_p^\tp(s)\hat{\mathbf C}\boldsymbol{\Phi}_p(t)$ with 
 \begin{equation}
 \begin{aligned}
  \hat{\mathbf{C}}=\mathop{\arg\min}\limits_{\mathbf{C}:\,\gamma_{\mathbf{C}}\in\mathcal C}\, \sum_{i=1}^{n}w_{i}\sum_{1 \le j \neq l \le m_{i}}\{\Gamma_{ijl}-\gamma_{\textbf{C}}(T_{ij},T_{il})\}^{2}+\lambda {J}(\gamma_{\mathbf{C}}),
 \end{aligned}\label{Mrep}
 \end{equation}
 where $\lambda$ is a tuning parameter that provides a trade-off between the fidelity to the data and the smoothness of the estimate. Numerical method to solve the constraint optimization \eqref{Mrep} is detailed in Appendix B.

Similar to \eqref{eq:mean-estimator}, the penalty term in \eqref{Mrep} is introduced to overcome excessive variability of an estimator when a large number of basis functions are required while the sample size is relatively small. It does not affect  the convergence rate of $\hat\gamma$ when the tuning parameter $\lambda$ is not too large. The choices of $p$ and $\lambda$ are discussed in Section \ref{subsec:computation}. For the choice of the weights $w_{i}$,  \cite{Zhang2016} discussed several weighing schemes, including the OBS scheme $w_i=1/\{\sum_{i=1}^{n}m_{i}(m_{i}-1)\}$ and the SUBJ scheme $w_i=1/\{nm_{i}(m_{i}-1)\}$. An optimal weighing scheme was proposed in \cite{Zhang2018a};  we refer to this paper for further details.

\section{Theory}
\label{sec:asyth}
As functional snippets are often sparsely recorded, in the sense that $m_i\leq m_0<\infty$ for all $i=1,\ldots,n$ and some $m_0>0$, in this section we focus on  theoretical analysis  tailored to this scenario. For simplicity, we assume that the number of observations for each trajectory is equal, i.e., $m_{i}=m$ for all $i=1,\ldots,n$. Note that under this assumption, the SUBJ and OBJ schemes coincide. The results for general number $m_{i}$ of observations and weight schemes can be derived in a similar fashion. We start with a discussion on the choice of basis functions and then proceed to study the convergence rates of the estimated mean and covariance functions.

\subsection{Analytic Basis}\label{subsec:basis}
While all complete orthonormal bases can be used for the proposed estimator in \eqref{eq:mean-estimator}, an analytic basis is preferred for the estimator in \eqref{Mrep}. For a clean presentation, we exclusively consider analytic bases $\boldsymbol{\Phi}=\{\phi_1,\ldots\}$ that work for both \eqref{eq:mean-estimator} and \eqref{Mrep}. In this paper,  a basis is called an analytic $(\alpha,\beta)$-basis if its basis functions are all analytic and satisfy the following property:  for some constants $\alpha,\beta\geq 0$, there exists a constant $C$ such that $\|\phi_k\|_\infty\leq Ck^\alpha$ and $\|\phi_k^{(r)}\|_{\ltwo}\leq Ck^{\beta r}$ for $r=1,2$ and all $k=1,\ldots$. Here, $\|\phi_k\|_\infty$ denotes the supremum norm of $\phi_k$, defined by $\sup_{t\in[0,1]}|\phi_k(t)|$, and $\phi_k^{(r)}$ represents the $r$th derivative of $\phi_k$. Throughout this paper, we assume that the basis $\boldsymbol\Phi=\{\phi_1,\ldots\}$ is an analytic $(\alpha,\beta)$-basis.

Different bases lead to different convergence rates of the approximation to $\mu_0$ and $\gamma_0$. For the mean function $\mu_0$, when using the first $q$ basis functions $\phi_1,\ldots,\phi_q$, the approximation error is quantified by 
\begin{equation*}
\mathcal E(\mu_0,\boldsymbol\Phi,q)=\left\Vert \mu_0-\sum_{k=1}^q a_k\phi_k \right\Vert_{\ltwo},
\end{equation*}
where we recall that $a_k=\int_0^1\mu_0(t)\phi_k(t)\diffop t$. The convergence rate of the error $\mathcal E(\mu_0,\boldsymbol\Phi,q)$, denoted by $\tau_q=\tau_q(\mu_0,\boldsymbol\Phi)$, signifies the approximation power of the basis $\boldsymbol\Phi$ for $\mu_0$. Similarly, the approximation error for $\gamma_0$ is measured by
\begin{equation*}
\mathcal E(\gamma_0,\boldsymbol\Phi,p)=\left\Vert \gamma_0-\sum_{k=1}^p\sum_{l=1}^p c_{kl}\phi_k\otimes\phi_l \right\Vert_{\ltwo},
\end{equation*}
where the $\ltwo$ norm of a function $\gamma(s,t)$ is defined by $\|\gamma\|_{\ltwo}=\{\int_0^1\int_0^1|\gamma(s,t)|^2\diffop s\diffop t\}^{1/2}$. The convergence rate of $\mathcal E(\gamma_0,\boldsymbol\Phi,p)$ is denoted by $\kappa_p=\kappa_p(\gamma_0,\boldsymbol\Phi)$. Below we discuss two examples of bases.

\begin{example}[Fourier basis]\label{ex:1}
Fourier basis functions, defined by $\phi_1(t)=1$, $\phi_{2k}(t)=\cos(2k\pi t)$ and $\phi_{2k+1}(t)=\sin(2k\pi t)$ for $k\geq 1$, constitute a complete orthonormal basis of $\ltwo(\tdomain)$ for $\tdomain=[0,1]$. It is also an analytic $(0,1)$-basis. When $\mu_0$ is periodic on $\tdomain$ and belongs to the Sobolev space $\mathscr H^r(\tdomain)$ (see Appendix A.11.a and A.11.d of \cite{Canuto2006} for the definition), 
then, according to Eq. (5.8.4) of \cite{Canuto2006}
one has $\tau_q=O(q^{-r})$. 
Similarly, if $\gamma_0$ is periodic and belongs to $\mathscr H^r(\tdomain^2)$, then $\kappa_p=O(p^{-r})$. 
When $\mu_0$ or $\gamma_0$ is not a periodic function, a technique called Fourier extension, briefly described in Appendix A, can be adopted to yield the same rate \citep{Adcock2014}. This technique is well studied in the field of 
 computational physics \citep{Boyd2002} and numerical analysis \citep{Huybrechs2010} as a tool to overcome the so-called Gibbs phenomenon \citep{Zygmund2003}, but seems not well explored in statistics yet. 
In Section \ref{sec:simu}, we numerically illustrate the application of this technique.
\end{example}

\begin{example}[Legendre polynomials]\label{ex:2} The canonical Legendre polynomial $P_k(t)$ of degree $k$ is defined on $[-1,1]$ by $$P_k(t)=\frac{1}{2^k k!}\frac{\diffop ^k}{\diffop t^k}\{(x^2-1)^k\}.$$ These polynomials are orthogonal in $\ltwo(-1,1)$. By a change of variable and normalization, they can be turned into an orthonormal basis of $\ltwo(\tdomain)$. The Legendre polynomials have numerous applications in  approximation theory and numerical integration; see \cite{Wang2011a} and references therein. One can show that the Legendre basis is an analytic $(1/2,1)$-basis. According to Eq (5.8.11) of \cite{Canuto2006}, one has $\tau_q=O(q^{-r})$ and $\kappa_p=O(p^{-r})$ when $\mu_0$ belongs to $\mathscr H^r(\tdomain)$ and $\gamma_0$ belongs to $\mathscr H^r(\tdomain^2)$, respectively.
\end{example}

\subsection{Mean function}\label{subsec:mean-theory}
For functional snippets, we shall assume that the observations from a subject scatter randomly in a subject specific time interval, whose length is $\delta$ and whose middle point is called the reference time in this paper. We further assume that the reference time $R_{i}$ of the $i$th subject are independently and identically distributed in the interval $[\delta/2,1-\delta/2]$, and the observed time points $T_{i1},\ldots,T_{im_{i}}$, conditional on $R_{i}$, are independently and identically distributed in the interval $[R_{i}-\delta/2,R_{i}+\delta/2]$.

To study the estimator $\hat{\mu}$, we make the following  assumptions.
 
 \begin{description}[labelwidth=1.2cm,leftmargin=1.4cm,align=left]
 	\item[\mylabel{cond:timept}{A.1}] There exist $0 < \mathfrak{c}_{1} \leq  \mathfrak{c}_{2} < \infty,$ such that the density $ f_{R}(s)$ of the reference time $R$ satisfies $\mathfrak{c}_{1} \leq  f_{R}(s) \leq  \mathfrak{c}_2,$ for any $s\in[\delta/2,1-\delta/2]$. There exist $0 < \mathfrak{c}_{3} \leq  \mathfrak{c}_{4} < \infty$, such that the  conditional density $f_{T|R}(t|s)$ of the observed time $T$ satisfies 
 	$\mathfrak{c}_{3} \leq f_{T|R}(t|s) \leq \mathfrak{c}_{4}$, for any given reference time $s\in[\delta/2,1-\delta/2]$ and $t\in [s-\delta/2,s+\delta/2]$.
 \end{description}
 
   \begin{description}[labelwidth=1.2cm,leftmargin=1.4cm,align=left]
	\item[\mylabel{cond:cov}{A.2}] $\expect\|X\|_{\ltwo}^2\leq \mathfrak{c}_5<\infty$ for some constant $\mathfrak{c}_5>0$.
\end{description}

\begin{description}[labelwidth=1.2cm,leftmargin=1.4cm,align=left]
	\item[\mylabel{cond:rho:q}{A.3}] 
	$q^{2\alpha+2}/n\rightarrow 0$ and $\rho / (n^{-1/2}q^{\alpha-4\beta-1/2})\rightarrow 0$.
\end{description}

Assumption A.1 requires the density of the reference time and conditional densities of the time points to be bounded away from zero and infinity. This also guarantees that the marginal probability density of the time points $T_{ij}$ is bounded away from zero and infinity.   Assumption A.2 is mild and assumption A.3 facilitates the convergence rate, where the dimension $q$ can grow with $n$.  In the sequel, we use 
$a_n\asymp b_n$ to denote  $0<\lim_{n\rightarrow \infty} a_n/b_n <\infty$.
\begin{theorem}\label{meanthm} If $\boldsymbol\Phi$ is a $(\alpha,\beta)$-basis, conditions (A.1)--(A.3) imply that 
	\begin{equation}
	\|\hat{\mu}-\mu_{0}\|_{\ltwo}^2=O_{P}\Big(\frac{q^{2\alpha+1}}{n}+\tau_q^2\Big),\label{murate}
	\end{equation} 
	where $\tau_q$ is the convergence rate of  
	$\mathcal E(\mu_0,\boldsymbol\Phi, q)$  defined in section 3.1. 
	\label{thm: mean1}
	\end{theorem}
We first note that,  under  condition A.3, the tuning parameter $\rho$  does not have direct impact on the asymptotic rate of $\hat{\mu}$.
We also observe that in \eqref{murate}, the term $q^{2\alpha+1}n^{-1}$ specifies the \emph{estimation error} using a finite sample, while $\tau_q$ is the deterministic \emph{approximation error} for using only the first $q<\infty$ basis functions. The latter term depends on the smoothness of $\mu_0$.
  Intuitively, it is  easier to approximate smooth functions with basis functions. For a given number of basis functions,  smoother functions generally yield smaller  approximation errors. 
As discussed in Example \ref{ex:1} and \ref{ex:2}, when $\mu_0$ belongs to the Sobolev space $\mathscr H^r(0,1)$, i.e.,  $\mu_0$ is $r$ times differentiable, 
we have  $\tau_q=O(q^{-r})$.  This leads to the following convergence rate,  using either the Fourier basis or the Legendre basis.
\begin{corollary}\label{cor:mean-r}
Suppose $\mu_0^{(r)}$ exists and satisfies $\|\mu_0^{(r)}\|_{\ltwo}<\infty$ for some $r\geq 1$. Assume conditions (A.1)--(A.3) hold. 
\begin{enumerate}[label=\textup{(\roman*)}]
	\item If $\boldsymbol\Phi$ is the Fourier basis and $\mu_0$ is periodic, then $\|\hat{\mu}- \mu_0 \|_{\ltwo}^2=\Op(n^{-2r/(2r+1)})$ with the choice $q\asymp n^{1/(2r+1)}$. 
	\item If $\boldsymbol\Phi$ is the Legendre basis, then $\|\hat{\mu}- \mu_0 \|_{\ltwo}^2=\Op(n^{-r/(r+1)})$ with the choice $q\asymp n^{1/(2r+2)}$. 
\end{enumerate}
\end{corollary}

For $r=2$ and a periodic function $\mu_0$, the convergence rate for $\hat{\mu}$ is $n^{-2/5}$ and $n^{-1/3}$, respectively, for the estimator based on Fourier and Legendre bases.  We can see that the convergence rate is faster for a Fourier basis. This is because, although they are both $(\alpha,\beta)$-bases,  $\alpha=1/2$ for the Legendre basis is larger than $\alpha=0$ for the Fourier basis. According to \eqref{murate}, a larger value of $\alpha$ leads to a slower rate. Indeed, $\alpha$ controls the growth rate of the extrema of basis functions. Fourier basis functions are uniformly bounded between $-1$ and $1$. In contrast, high-order Legendre basis functions tend to have large extrema that amplify variability. This limits the number of basis functions for estimation and thus causes a slower convergence rate for the Legendre basis. When $\mu_0$ is nonperiodic, the classic Fourier basis suffers from the so-called Gibbs phenomenon which, however, can be substantially alleviated by the Fourier extension technique; see Section \ref{sec:simu} and Appendix A for more details.

\subsection{Covariance function}\label{subsec:cov-theory}

In Section \ref{sec:meth} we assumed $\gamma_0$ to reside in a $\tdomain_\delta$-identifiable family $\mathcal{C}$ in order to meet a basic criterion of  identifiability. To study the asymptotic properties of the covariance estimator, we require the family $\mathcal C$ to satisfy an additional regularity condition as described below.
Let $\mathcal{F}$ be the space of real-valued functions defined on $\tdomain^2$  endowed with the product topology. In this topology, a sequence of functions $\{f_k\}$ converges to a limit $f$ if and only if 
\begin{equation*}
	\begin{aligned}
	\lim\limits_{k\rightarrow\infty}f_{{k}}(s,t)=f(s,t),\quad \text{for\ all} \ (s,t)\in \tdomain^{2}.
	\end{aligned}
\end{equation*}
\begin{definition}\lin{
	A subset $\mathcal{S}$ of $\mathcal{F}$ is called a \emph{bounded sequentially compact} (BSC) family if every sequence $\{f_k\}\subset \mathcal{S}$ that is bounded in the $L^2$ norm, i.e., $\sup_k\|f_k\|_{\ltwo}<\infty$,  has a subsequence converging to a limit in $\mathcal{S}$ in the product topology.}
\end{definition}
\lin{The BSC concept is closely related to the topological concept ``sequential compactness''. Specifically,  a subset $\mathcal S\subset \mathcal F$ is sequentially compact if every sequence in $\mathcal S$ has a subsequence that converges to a limit in $\mathcal S$ in the topology of $\mathcal F$. In particular, sequential compactness is stronger than BSC and thus implies the latter. However, when the subset $\mathcal S$ is uniformly bounded in the $L^2$ norm, i.e., $\sup_{f\in\mathcal S}\|f\|_{\ltwo}<\infty$, then the two concepts coincide. If all functions in $\mathcal S$ are bounded by a common constant and are Lipschitz continuous with a common Lipschitz constant, then $\mathcal S$ is a BSC family. To see this, we note that such family is locally equicontinuous, and also the set $\{f(s,t):f\in \mathcal{S}\}$ is bounded for all $s,t\in \tdomain$. Then the claim follows from the Arzel\`a--Ascoli Theorem \citep[Chapter 7,][]{Remmert1997}. Also, the product of two BSC families of which the functions are uniformly bounded in the $\ltwo$ norm is also a BSC family. This property is useful for constructing new BSC families from existing ones; see Example \ref{ex:8} for an illustration. The following proposition, of which the proof is trivial or already discussed in the above and thus is omitted, summarizes the aforementioned properties of BSC and its connections to sequential compactness.}
\begin{proposition}\label{prop:BSC}\lin{
	Let $\mathcal F$ be the collection of real-valued functions defined on $\tdomain^2$ and endowed with the product topology. Let $\mathcal S$ be a subset of $\mathcal F$.
	\begin{enumerate}[label=\textup{(\roman*)}] 
		\item\label{prop:BSC-i} If $\mathcal S$ is sequentially compact, then it is a BSC family.
		\item\label{prop:BSC-ii} If $\mathcal S$ is a BSC family and $\sup_{f\in \mathcal S}\|f\|_{\ltwo}<\infty$, then $\mathcal S$ is sequentially compact.
		\item\label{prop:BSC-iii} If $\sup_{f\in\mathcal S}\|f\|_\infty<\infty$ and $\sup_{f\in\mathcal S}\sup_{x\neq y}|f(x)-f(y)|/\|x-y\|_2<\infty$, then $\mathcal S$ is a BSC family.
		\item\label{prop:BSC-iv} If both $\mathcal S$ and  $\mathcal Q\subset \mathcal F$ are sequentially compact, then the family $\mathcal S\mathcal Q=\{fg:f\in\mathcal S,g\in\mathcal Q\}$ is also sequentially compact. Consequently, if $\mathcal S$ and  $\mathcal Q$ are BSC families satisfying $\sup_{f\in \mathcal S\cup\mathcal Q}\|f\|_{\ltwo}<\infty$, then $\mathcal S\mathcal Q$ is also a BSC family.
	\end{enumerate}}
\end{proposition}

 The following examples utilize Proposition \ref{prop:BSC} to exhibit BSC families and illustrate the abundance of BSC families.

\begin{example}[Bounded Sobolev sandwich families] \label{ex:3}
Let $M_1,M_2>0$ be fixed but potentially arbitrarily large constants. Let $\mathcal{S}(M_1,M_2)$ be the subfamily of the $\tdomain_\delta$-identifiable family introduced in Example \ref{ex:ssf} such that, if $f\in \mathcal{S}(M_1,M_2)$ then $\|f\|_\infty\leq M_1$ and the Lipschitz constant of $f$ is bounded by $M_2$, where the Lipschitz constant of $f$ is defined as  $\sup_{x\neq y} |f(x)-f(y)|/\|x-y\|_2$.  By Proposition \ref{prop:BSC}\ref{prop:BSC-iii}, $\mathcal S(M_1,M_2)$ is a BSC family.
\end{example}
\begin{example}[BSC semiparametric families] \label{ex:8}
	\lin{
	Let $\mathcal H$ be a family of covariance functions indexed by a parameter $\theta$ in a compact space $\Theta\subset\real^d$ for some $d>0$. If each $f_\theta\in\mathcal H$ is Lipschtiz continuous and has a Lipschitz constant continuous in $\theta$, i.e., $|f_\theta(s)-f_\theta(y)|\leq \ell_\theta \|x-y\|_2$ for $x,y\in\real^2$ and $\ell_\theta$ is continuous in $\theta$ on $\Theta$, then $\mathcal{H}$ is a BSC family. To see this, we note that the continuity and compactness of $\tdomain$ and $\Theta$ imply that $\|f\|_\infty\leq M_1$ and $\sup_{x\neq y}|f(x)-f(y)|/\|x-y\|_2\leq M_2$, for some constants $M_1,M_2>0$. Then the claim follows from Proposition \ref{prop:BSC}\ref{prop:BSC-iii}.  Similar reasoning shows that the family $\mathcal G=\{g\otimes g:$ $\|g\|_\infty <M_3$ and $\|g^\prime\|_\infty <M_4\}$ is a BSC family of which functions are uniformly bounded in the $\ltwo$ norm, where $(g\otimes g)(s,t)=g(s)g(t)$, and $M_3,M_4>0$ are constants. According to Proposition \ref{prop:BSC}\ref{prop:BSC-iv},  the semiparametric family $\{gh:g\in\mathcal G,h\in\mathcal H\}$ is a BSC family.}
\end{example}

  We also note that the construction in the above example can be used to derive BSC subfamilies of the $\tdomain_{\delta}$-identifiable families introduced in Example \ref{ex:af} and \ref{ex:lp}. All of these BSC subfamilies clearly contain countless covariance functions of infinite rank. In the sequel we shall assume the family $\mathcal C$ under consideration is a BSC family. The above examples suggest that this regularity condition, essentially controlling the complexity of the family, holds for any family of functions that are collectively bounded and Lipschitz continuous, and thus is mild.
Formally, we shall assume the following conditions.
\begin{description}[labelwidth=1.2cm,leftmargin=1.4cm,align=left]
\item[\mylabel{cond:bsc}{B.1}]  The covariance function $\gamma_0$ belongs to a $\tdomain_\delta$-identifiable BSC family $\mathcal{C}$.
\end{description}
  \begin{description}[labelwidth=1.2cm,leftmargin=1.4cm,align=left]
	\item[\mylabel{cond:Xmoment}{B.2}]  The random function $X$ satisfies  $\expect \|X\|^4_{\ltwo}<\infty$. 
\end{description}
  \begin{description}[labelwidth=1.2cm,leftmargin=1.4cm,align=left]
	\item[\mylabel{cond:curve}{B.3}] $p^{8\alpha+4}/n\rightarrow 0$  and $\lambda/(n^{-1/2}p^{2\alpha-4\beta-3/2})\rightarrow 0$ as $n\rightarrow \infty$.
\end{description}
Since a $\tdomain_\delta$-identifiable BSC family, such as $\mathcal{S}(M_1,M_2)$ in Example \ref{ex:3},  may contain covariance functions of infinite rank, the theory developed below  applies to functional snippets of infinite dimension.  
To avoid the  entanglement with the error from mean function estimation, we shall assume that $\mu_0$ is known in the following discussion, noting that the case that $\mu_0$ is unknown can also be covered but requires  a much more involved presentation and tedious technical details and thus is not pursued here. \lin{The following result establishes the convergence rate of the proposed estimator for any class of analytic $(\alpha,\beta)$-bases.}
\begin{theorem}\label{thm:2} If $\boldsymbol\Phi$ is an analytic $(\alpha,\beta)$-basis, under assumptions A.1, B.1--B.3 and $m\geq 2$, we have 
	\begin{equation}\label{eq:gamma-rate1}
		\|\hat\gamma-\gamma_{0}\|_{\ltwo}^2=O_{P}\Big(\frac{p^{4\alpha+2}}{n}+\kappa_p^2\Big),	
	\end{equation}
	where $\kappa_p$ is the convergence rate of  
	$\mathcal E(\gamma_0,\boldsymbol\Phi,p)$  defined in section 3.1. 
\end{theorem}
Note that, with the condition B.3 on $\lambda$, the tuning parameter $\lambda$ does not  impact  the asymptotic rate of $\hat{\gamma}$. As in the case of the mean function, the rate in \eqref{eq:gamma-rate1} contains two components, the estimation error $p^{4\alpha+2}n^{-1}$ stemming from the finiteness of the sample, and the approximation bias $\kappa_p$ attributed to the finiteness of the number of basis functions being used in the estimation. When the Fourier basis or Legendre basis is used, we have the following convergence rate for $r$ times differentiable covariance functions.

\begin{corollary}\label{cor:cov}
	Suppose assumptions A.1 and B.1--B.3 hold, $m\geq 2$, and $\gamma_0$ belongs to the Sobolev space $\mathscr H^r(\tdomain^2)$ for some $r\geq 1$. 
\begin{enumerate}[label=\textup{(\roman*)}]
	\item If $\boldsymbol\Phi$ is the Fourier basis and $\gamma_0$ is periodic, then with $p\asymp n^{1/(2r+2)}$, one has $\|\hat{\gamma}-\gamma_{0}\|_{\ltwo}^2=O_{P}\big(n^{-r/(r+1)}\big)$.
	\item If $\boldsymbol\Phi$ is the Legendre basis, then  with $p\asymp n^{\min\{1/(2r+4),1/8\}}$, one has 
	$\|\hat{\gamma}-\gamma_{0}\|_{\ltwo}^2=O_{P}\big(n^{-\min\{r/(r+2),r/4\}}\big)$.
\end{enumerate}
\end{corollary}

When $r=2$ and $\gamma_0$ is periodic, the convergence rate for $\hat\gamma$ is $n^{-1/3}$ and $n^{-1/4}$, respectively, for the estimator  based on the Fourier and Legendre bases. The reason behind this observation is similar to the case of the mean function: high-order Legendre basis functions tend to have large extrema that amplify variability, which limits the number of basis functions for estimation and leads to a relatively slower rate. For a nonperiodic $\gamma_0$, the Fourier extension technique can also be used to alleviate the Gibbs phenomenon suffered by the Fourier basis; see Appendix A for more details.

The above corollary, as well as Corollary \ref{cor:mean-r}, shows that the proposed analytic basis expansion approach automatically adapts to the smoothness of $\mu_0$ and $\gamma_0$. In particular, when $\mu_0$ or $\gamma_0$ is smooth, i.e., has infinite order of differentiability, our estimators enjoy a near-parametric rate. This contrasts with the local polynomial smoothing method and the B-spline basis approach, for which the convergence rate is limited by the order of the polynomials or B-spline basis functions  used in the estimation, even when $\mu_0$ or $\gamma_0$ might have a higher order of smoothness. In practice, it is not easy to determine the right order for these methods, since the mean and covariance functions and their smoothness are unknown.

\begin{remark}\lin{
	\cite{Delaigle2019} adopted a two-stage procedure to estimate the covariance function, where in the first stage a pilot estimate is constructed only in the diagonal region $\mathcal T_\delta$. This pilot estimate can be obtained through any 2D smoother, such as PACE \citep{Yao2005a}.  At the second stage, this pilot estimate is  numerically extrapolated by basis expansion (without penalty) to arrive at an estimate of the entire covariance function.
	The advantage of this two-stage approach is that the convergence rate of the final estimator is immediately available and inherited from the convergence rate of the pilot estimator, since the basis approximation error is negligible when a sufficiently large number of basis functions are used. The drawback is that, the convergence rate \citep[Theorem 2,][]{Delaigle2019} is then limited by the pilot estimate. For instance, if a local linear smoother is adopted to produce the pilot estimate as in PACE, then the convergence rate of the final estimator is the same as  the local linear smoother, even though the true covariance function might have higher-order degree of smoothness. In contrast, our approach  estimates the basis coefficients directly from the data through penalized least squares, and thus is able to automatically exploit the high-order smoothness of the true covariance. This also allows us to establish an explicit convergence rate without reference to a pilot estimate. Note that the convergence theories in \cite{Delaigle2019} and our paper are based on different and  incomparable sets of conditions, and thus do not imply each other.}
\end{remark}
\begin{remark}\lin{One of the theoretical novelties  in \cite{Delaigle2019} is the  approximation  error bound in its  Theorem 2 to quantify the approximation errors, when the covariance function is not identifiable. This is a very appealing feature and upon the suggestion of a referee we address the case that $\gamma_0$ is not in the  $\tdomain_\delta$-identifiable BSC family $\mathcal{C}$ but can be well approximated by a member of the family. }
\end{remark}

Specifically, let $\tilde{\gamma}\in\mathcal{C}$ and assume $\|\gamma_{0}-\tilde{\gamma}\|_{\hsnorm}\leq \eta$ for some constant $\eta\geq 0$. Below we show that, when the model of $\gamma_{0}$ is misspecified, the estimation quality also depends on the degree of misspecification that is quantified by $\DD$.
\begin{theorem}\label{thm:3}\lin{Suppose that $\boldsymbol\Phi$ is an analytic $(\alpha,\beta)$-basis, $m\geq 2$, and assumptions A.1 and B.2--B.3 hold. If there exists \lin{$\tilde{\gamma}$ in the $\tdomain_\delta$-identifiable  BSC family $\mathcal{C}$ such that $\|\tilde{\gamma}-\gamma_0\|_{L^2}\leq\DD$}, then we have 
	\begin{equation}\label{eq:gamma-rate}
	\|\hat\gamma-\gamma_{0}\|_{\ltwo}^2=O_{P}\Big(\frac{p^{4\alpha+2}}{n}+\kappa_p^2 + \DD^2\Big),	
	\end{equation}
	where $\kappa_p$ is the convergence rate of  
	$\mathcal E(\gamma_0,\boldsymbol\Phi,p)$  defined in Section 3.1.}
\end{theorem}
\begin{corollary}\lin{
	Suppose that the conditions of Theorem \ref{thm:3} hold, and $\gamma_0$ belongs to the Sobolev space $\mathscr H^r(\tdomain^2)$ for some $r\geq 1$. 
	\begin{enumerate}[label=\textup{(\roman*)}]
		\item If $\boldsymbol\Phi$ is the Fourier basis and $\gamma_0$ is periodic, then with $p\asymp n^{1/(2r+2)}$, one has $\|\hat{\gamma}-\gamma_{0}\|_{\ltwo}^2=O_{P}\big(n^{-r/(r+1)}+\DD^2\big)$.
		\item If $\boldsymbol\Phi$ is the Legendre basis, then  with $p\asymp n^{\min\{1/(2r+4),1/8\}}$, one has 
		$\|\hat{\gamma}-\gamma_{0}\|_{\ltwo}^2=O_{P}\big(n^{-\min\{r/(r+2),r/4\}}+\DD^2\big)$.
	\end{enumerate}	}
\end{corollary}

\begin{remark}
	\lin{
	Although for simplicity we focus on functional snippets where each observed trajectory consists of only a single snippet of equal width, without any modification, the  proposed estimation procedure in Section \ref{sec:meth} is applicable to more general cases, including the case that snippets have random and different widths of span and/or the case that each trajectory is composed by multiple pieces of snippets. The theory developed in this section can also accommodate such cases with slight modification of condition A.1, as follows. For random width $\delta$, we require $\prob(\delta > \delta_0)>0$ for some constant $\delta_0\in(0,1)$. To model multiple pieces of snippets per trajectory, one can introduce multiple reference time points, one for each piece, i.e., for the $j$th piece within a single trajectory, there is a reference time $R^{(j)}$. Without loss of generality, we assume $R^{(1)}<R^{(2)}<\cdots < R^{(S)}$, where the (potentially random) quantity $S\geq 1$ denotes the number of snippets per trajectory. Then our theories are still valid if the condition A.1 holds for the first and last pieces, i.e. for the pieces indexed by the reference time points $R^{(1)}$ and $R^{(S)}$.}
\end{remark}

\section{Numerical Studies}\label{sec:simu}

\subsection{Computational Details}\label{subsec:computation}

To compute the estimator in \eqref{eq:mean-estimator}, one needs to determine a set of basis functions. We recommend the Fourier basis, for it is often computationally stabler than Legendre polynomials and other polynomial bases. To handle nonperiodic data, we incorporate the technique of Fourier extension described in Appendix A, and thus need to select an additional tuning parameter, the extension margin $\zeta$.  Through extensive numerical experiments, we found that the results are often not sensitive to $\zeta$ when it is not too large and too small. As a rule-of-thumb, we recommend $\zeta$ to be one tenth of the span of the study. If computational capacity allows, a data-driven value for $\zeta$ can also be selected via cross-validation. 

To select the other two tuning parameters $q$ and $\rho$, we adopt a $K$-fold cross-validation procedure with $K=5$, as follows. Let $\Xi$ and $\Theta$ be sets of candidate values for $q$ and $\rho$, respectively. We choose a pair $(q,\rho)\in \Xi\times \Theta$ that minimizes the validation error $$\mathrm{CV}(q,\rho)=\sum_{k=1}^K\sum_{i \in \mathcal{P}_k}\sum_{j=1}^{m_i}(Y_{ij}-\hat{\mu}_{-k}^{q,\rho}(T_{ij}))^2,$$
where $\mathcal{P}_1,\ldots,\mathcal{P}_K$ form a roughly even partition of $\{1,\ldots,n\}$, and $\hat{\mu}_{-k}^{q,\rho}$ is the estimator for $\mu_0$ with $q$ basis functions and the penalization parameter $\rho$ and only using subjects with indices in $\{1,\ldots,n\}\backslash\mathcal{P}_k$.

To compute the covariance estimator in \eqref{Mrep}, we shall choose a BSC family $\mathcal{C}$, and for our simulation studies we adopt the family exhibited in Example \ref{ex:3} with large values of $M_1=100$ and $M_2=100$, since it is a large family that allows us to reduce model bias while identify the covariance function. The matrix $\mathbf{C}$ in \eqref{Mrep} is positive definite, which renders the optimization difficult. We tackle the issue of positive definiteness via  geometric Newton method by realizing that the space of symmetric positive definite matrices is a Riemannian manifold when it is endowed with the easy-to-compute Log-Cholesky metric \citep{Lin2019+a}; see Appendix B for details. 
\lin{In contrast,  \cite{Delaigle2019} reparameterized $\mathbf{C}$ by its Cholesky factor $\mathbf{B}$, which is a lower triangular matrix satisfying $\mathbf{C}=\mathbf{B}\mathbf{B}^\top$, and then turned it into a nonconstrained optimization problem. However, this approach is observed to suffer from numerical instability, as the Cholesky decomposition $\mathbf{C}=\mathbf{B}\mathbf{B}^\top$ is not unique.}

For the other two parameters $p$ and $\lambda$, we adopt the following selection procedure. Let $\Xi$ and $\Theta$ be sets of candidate values for $p$ and $\lambda$, respectively. We choose a pair $(p,\lambda)\in \Xi\times \Theta$ that minimizes the error 
\begin{equation}\label{eq:cov-cv}
\mathrm{Err}(p,\lambda)=\sum_{\stackrel{1\leq j,l\leq G}{|t_j-t_l|\leq \hat\delta}}\{\check{\gamma}(t_j,t_l)-\hat{\gamma}^{p,\lambda}(t_j,t_l)\}^2,
\end{equation}
where $t_1,\ldots,t_G$ are $G$ equally spaced points on the domain $\tdomain$ for some $G>0$, $\hat\delta=\max\{|T_{ij}-T_{il}|: 1\leq i\leq n, 1\leq j,l\leq m_i\}$ is an estimate for $\delta$, $\check{\gamma}$ is a pilot estimate for $\gamma_0$ on the diagonal region $\tdomain_{\hat\delta}$ and can be computed by PACE \citep{Yao2005a,fdapace}, and $\hat{\gamma}^{p,\lambda}$ is the estimator with $p$ basis functions and the penalization parameter $\lambda$. Unlike the selection of $q$ and $\rho$ for estimating the mean function, we use \eqref{eq:cov-cv} instead of the cross-validation error that would be computed from the raw observations $\Gamma_{ijl}=\{Y_{ij}-\hat\mu(T_{ij})\}\{Y_{il}-\hat\mu(T_{il})\}$. This is because the raw observations are too noisy and often result in substantial  variability in the cross-validation procedure. In contrast, by utilizing the PACE estimate, we not only denoise the raw observations but also better leverage the information available in the diagonal region through the equally spaced grid $(t_1,\ldots,t_G)$. In \cite{Delaigle2019}, the pilot estimator $\check{\gamma}$ is used to directly estimate the basis coefficients, while we use it for selecting tuning parameters. The simulation studies in the next subsection demonstrate that our strategy is favored in most cases. 
\subsection{Monte Carlo Simulations}

We now illustrate the numerical performance of the proposed approach using the Fourier basis. 
For the mean function, we consider two scenarios, $\mu_1(t) = \sum_{k=1}^9 (-1)^k1.2^{-k}\phi_k(t)$ and $\mu_2(t)=2t$. The former is a periodic function while the latter is nonperiodic. For the covariance function, we consider the following cases:
\begin{itemize}
	\item[I.] the periodic covariance function $\gamma_1(s,t)=(\phi_1(s),\ldots,\phi_5(s))\cdot \mathbf{C} \cdot (\phi_1(t),\ldots,\phi_5(t))^\tp$, where $\mathbf C$ is a $5\times 5$ matrix with $\textbf{C}=\{c_{kl}\}$ and $c_{kl}=2^{-|k-l|-5/2}$ if $k\neq l$ and $1.5^{1-k}$ if $k=l$;
	\item[II.]\lin{the nonperiodic and nonsmooth covariance function $\gamma_2(s,t)$ that is determined by the correlation function $e^{-|s-t|^2}$ and the variance function $v(t)=\{1+\int_0^t (1+\lfloor4.5x\rfloor)\diffop x\}/\sqrt{2}$, \[
	\gamma_2(s,t)=\sqrt{v(s)v(t)/2}e^{-|s-t|^2},
	\]
	where $\lfloor4.5x\rfloor$ denotes the integer part of $4.5x$;}
	\item[III.] \lin{the periodic covariance function $\gamma_3$ that is the same as $\gamma_1$ except that the dimension of $\mathbf C$ is increased to $30\times 30$;}
	\item[IV.] \lin{the one-rank covariance function $\gamma_4(s,t)=0.4\varphi(s)\varphi(t)$ with $\varphi(t)=0.3f_{0.3,0.05}(t)+0.7f_{0.7,0.05}(t)$, where $f_{a,b}$ is the probability density of the normal distribution with mean $a$ and standard deviation $b$.}
\end{itemize}
\lin{Covariance functions in the cases I, III and IV fall into the family $\mathcal C$ we chose in Section \ref{subsec:computation}, while the one in the case II is not, since the function $\gamma_2$ is nonsmooth and falls outside our chosen family $\mathcal{C}$. The covariance functions in the last two cases, $\gamma_3$ and $\gamma_4$, although have different ranks, require a large number of Fourier basis functions for a good approximation. All of the last three covariance functions represent challenging cases for our approach.}

In the evaluation of the performance for the mean function, the covariance function is fixed to be $\gamma_1$, while the mean function is fixed to be $\mu_1$ in the evaluation of the covariance function. This strategy avoids the bias  from covariance to influence the estimation of the mean function, and vice versa. 

The estimation quality is measured by the empirical mean integrated squared error (MISE) based on $N=100$ independent simulation replicates. For the mean estimator $\hat\mu$, the MISE is defined by 
\[\text{MISE}=\frac{1}{N}\sum_{k=1}^{N}\int\{\hat\mu_{k}(t)-\mu_0(t)\}^2\diffop t, \]
and for the covariance estimator $\hat\gamma$, it is defined by  
\[\text{MISE}=\frac{1}{N}\sum_{k=1}^{N}\iint\{\hat\gamma_{k}(s,t)-\gamma_0(s,t)\}^2\diffop s\diffop t,\]
where $\hat\mu_k$ and $\hat\gamma_k$ are estimators in the $k$th simulation replicate. The tuning parameters $q$, $p$, $\rho$, $\lambda$ and the extension margin (see Appendix A) are selected by the procedures described in Section \ref{subsec:computation}.

In all replicates, the reference time $R_i$ are sampled from a uniform distribution on $[\delta/2,1-\delta/2]$. The number of observations $m_{i}$ are independent and follow the distribution $2+\text{Poisson(3)}$.  The measurement noise variables $\varepsilon_{ij}$  are i.i.d. sampled from a Gaussian  distribution with mean zero and variance $\sigma^{2}$, where the noise level $\sigma^2$ is set to make the signal-to-noise ratio $\expect \|X-\mu_0\|_{\ltwo}^{2}/\sigma^{2}=4$. We consider three sample sizes, $n=50, 150, 500$, and two different values of $\delta$, $\delta = 0.25, 0.75$,  representing short snippets and long snippets, respectively. 

 The results are summarized in Table \ref{tab:mean} and \ref{tab:cov} for the mean and covariance functions, respectively. As expected, in all settings, the performance of the estimators improves as $n$ or $\delta$  gets larger. We also observe that if the function to be estimated is periodic, like $\mu_1$ and $\gamma_1$,  it is better not to use Fourier extensions. However, if the function is nonperiodic, like $\mu_2$ and $\gamma_2$, then the estimators with Fourier extension considerably outperform those without the extension, especially for the mean function or when the sample size is large. This demonstrates that Fourier extension is a rather effective technique that complements the Fourier basis for nonparametric smoothing, and might deserve further investigation in the framework of statistical methodology. 
 
 As a comparison, we follow the description of \cite{Delaigle2019} to implement their estimator for the covariance function, where the pilot estimate is obtained by local linear smoothing in the region $\tdomain_{\delta}$.
 \lin{We also compare our method with the estimator $\hat\gamma_{DP}$ of \cite{Descary2019recovering} and and the one $\hat\gamma_{ZC}$ of \cite{Zhang2018}. These two estimates are obtained by extrapolating the raw covariance function in the region $\tdomain_{\delta}$ using matrix completion techniques. In our implementation of these estimators, we replace the raw covariance function in $\tdomain_{\delta}$ with the smooth pilot estimate obtained by local linear smoothing, as we find that this way substantially improves the quality of $\hat\gamma_{DP}$ and $\hat\gamma_{ZC}$. In particular, the obtained  matrix completion estimators are also relatively smooth, as demonstrated by the estimators $\hat\gamma_{DP}$ and $\hat\gamma_{ZC}$ in Figures \ref{fig:cov1-d25-n150}--\ref{fig:cov4-d25-n150}. This contrasts with the implementation of these estimators in \cite{Delaigle2019}, as shown in their Figure 4, in which the raw covariance is used and the yielded estimators exhibit significant roughness and variability. For both $\hat\gamma_{DP}$ and $\hat\gamma_{ZC}$, we set the resolution (the number of grid) to $d=51$. We then follow the description of  \cite{Descary2019recovering} and \cite{Zhang2018} to set other tuning parameters required to compute $\hat\gamma_{DP}$ and $\hat\gamma_{ZC}$, except the rank. For $\hat\gamma_{ZC}$, although \cite{Zhang2018} suggested a random subsampling cross-validation to determine the rank, no details are given for the number of random splits. As of  \cite{Delaigle2019}, we set the rank of $\hat\gamma_{ZC}$ to the true rank for the cases I, III and IV, and to the maximum allowed rank for the case II. Therefore, the $\hat\gamma_{ZC}$ we computed here is an oracle estimator, as it uses the practically inaccessible true rank. For $\hat\gamma_{DP}$, although \cite{Descary2019recovering} provided a graphical strategy to select the rank, it is not suitable for extensive simulation studies. To determine the rank for $\hat\gamma_{DP}$, we apply their graphical strategy to five random samples for each of the simulation settings. This yields five selected ranks, and we use the median  $r_{DP}$ of them for all simulated replicates. However, we observe that the estimator $\hat\gamma_{DP}$ is sensitive to the rank, and the value $r_{DP}$ might not be suitable for all replicates. To compensate this, we also compute the estimators $\hat\gamma_{DP,-1}$ corresponding to the rank  $r_{DP}-1$ (if  $r_{DP}>1$) and $\hat\gamma_{DP,+1}$ corresponding to the rank  $r_{DP}+1$, respectively, and set the MISE of $\hat\gamma_{DP}$ to the minimum of the MISE of $\hat\gamma_{DP}$, $\hat\gamma_{DP,-1}$ and $\hat\gamma_{DP,+1}$.}
 
\lin{Table \ref{tab:cov} summarizes the numerical results, while Figures \ref{fig:cov1-d25-n150}--\ref{fig:cov4-d25-n150}, depicting the estimated covariance function corresponding to the median MISE for each method when $\delta=0.25$ and $n=150$, provide visual comparison of the four methods. Based on the results, we make the following observations for each case.
\begin{itemize}
	\item[I.] The proposed method substantially outperforms the others in all settings, which is also evidenced by Figure \ref{fig:cov1-d25-n150}. This hardly surprising, since the true covariance function $\gamma_1$ favors our method.
	\item[II.] All methods have similar performance, which is also supported by Figure \ref{fig:cov2-d25-n150}. Our method is slightly favored when the sample size is small, while the matrix completion methods have slightly better performance when the sample size is large.
	\item[III.] Our method has superior performance in all settings. This could attribute to the fact that the true covariance function is still representable by a finite number of Fourier basis functions, although the number of Fourier basis functions in such representation is large.
	\item[IV.] Since this is a rather challenging case, all methods have close but deteriorated performance, and there is no clear winner. Figure \ref{fig:cov4-d25-n150} suggests that all methods are able to capture the four modes of the true covariance function, but none of them can get close to the magnitude of these modes.
\end{itemize}}
In summary, the performance of our method dominates the one of the DHHK estimator which also uses Fourier basis expansion. \lin{As previously mentioned, DHHK uses Fourier basis expansion as a purely numerical device to extrapolate a pilot estimate on the diagonal region obtained by a smoothing method, while ours directly estimates the basis coefficients from the data based on penalized least squares. We find that the approach of \cite{Delaigle2019} tends to produce excessive variability, especially in the off-diagonal region. In addition, these two methods adopt different strategies to handle nonperiodicity. Specifically, DHHK adds the identity function to the basis system, while ours uses the numerically well established Fourier extension technique. From the case II and Figure \ref{fig:cov2-d25-n150}, our strategy seems more promising. These primary differences between DHHK and our method might explain why the latter has superior performance relative to the former in almost all cases and settings. Comparing to the matrix completion methods, the proposed method seems preferred when the sample size is small or a small number of basis functions are sufficient for a good approximation to the true covariance function. In other scenarios, like the cases II and IV with a large sample size, the performance of our estimator is nearly as good as the matrix completion methods.}
 
\begin{table}[t]
	\caption{MISE of the proposed estimator for the mean function. The MISE and their Monte Carlo standard errors in this table are scaled by 10 for a clean presentation. FE refers to Fourier basis with Fourier extension, while NFE refers the basis without the extension.\label{tab:mean}}
	\begin{center}
		\renewcommand{\arraystretch}{1.2}
		\begin{tabular}{|c|c|c|c|c|c|c|c|}
			\hline
			&$n$ &\multicolumn{2}{|c|}{$50$} &\multicolumn{2}{|c|}{$150$} & \multicolumn{2}{|c|}{$450$} \tabularnewline
			\cline{2-8}
			&$\delta$&0.25&0.75&0.25&0.75&0.25&0.75\\
			\hline
			\hline
			\multirow{2}{*}{$\mu_{1}$} &FE&3.27(3.19) &1.94(1.27) &1.10(0.59) &0.68(0.44)&0.35(0.20)&0.25(0.16)   \\
			\cline{2-8}
			&NFE&3.18(3.09)&1.89(1.26)&1.05(0.57)&0.65(0.41)&0.34(0.19)&0.23(0.16)\\
			\hline
			\multirow{2}{*}{$\mu_{2}$}   &FE&2.74(3.20)&1.53(0.95)&0.99(0.73)&0.66(0.50)&0.42(0.28)&0.27(0.20)\\
			\cline{2-8}
			&NFE&3.20(3.03)&2.00(0.75)&1.61(0.60)&1.20(0.37)&0.89(0.31)&0.77(0.21)\\
			\hline
		\end{tabular}
	\end{center}
\end{table}
\begin{table}
	\caption{MISE of the proposed estimator for the covariance function. The MISE and their Monte Carlo standard errors in this table are scaled by 10 for a clean presentation. LWZ refers to the proposed Fourier basis expansion approach with Fourier extension, while LWZNE refers the one without the extension. DHHK refers to the estimators of \cite{Delaigle2019}. DP refers to the estimators of \cite{Descary2019recovering}. ZC refers to the estimators of \cite{Zhang2018}.\label{tab:cov}}
	\begin{center}
		\renewcommand{\arraystretch}{1.2}
		\resizebox{\columnwidth}{!}{%
			\begin{tabular}{|c|c|c|c|c|c|c|c|}
				\hline 
				& $n$ & $\delta$  & LWZ & LWZNE & DHHK & DP & ZC \tabularnewline
				\hline 
				\multirow{6}{*}{$\gamma_1$} & \multirow{2}{*}{50} & 0.25 & 6.40(3.90) & 6.32(3.85) & 10.65(4.10) & 9.64(5.12) & 8.22(4.66)\tabularnewline
				\cline{3-8}
				& & 0.75 & 6.04(3.17) & 5.81(3.91) & 9.46(3.59) & 7.58(4.89) & 8.11(3.93)\tabularnewline
				\cline{2-8}
				& \multirow{2}{*}{150} & 0.25 & 3.60(2.20) & 3.44(2.05) & 8.94(3.13) & 8.53(5.17) & 7.71(4.21)\tabularnewline
				\cline{3-8}
				& & 0.75 & 3.38(2.01) & 3.11(2.25) & 5.91(2.99) & 4.42(3.36) & 5.64(3.18)\tabularnewline
				\cline{2-8}
				& \multirow{2}{*}{450} & 0.25 & 2.22(1.21) & 2.09(1.12) & 7.61(3.06) & 7.67(4.25) & 6.68(3.01)\tabularnewline
				\cline{3-8}
				& & 0.75 & 1.52(1.39) & 1.42(1.38) & 3.53(2.65) & 2.31(1.79) & 3.11(2.18)\tabularnewline
				\cline{1-8}
				\multirow{6}{*}{$\gamma_2$} & \multirow{2}{*}{50} & 0.25 & 5.40(3.50) & 5.51(3.73) & 7.88(4.89) & 6.35(3.49) & 5.83(2.82)\tabularnewline
				\cline{3-8}
				& & 0.75 & 4.09(3.03) & 4.16(3.13) & 4.90(2.93) & 4.41(3.10) & 4.13(2.75)\tabularnewline
				\cline{2-8}
				& \multirow{2}{*}{150} & 0.25 & 2.65(1.47) & 2.74(1.43) & 4.32(4.15) & 3.73(4.40) & 3.07(3.33)\tabularnewline
				\cline{3-8}
				& & 0.75 & 2.35(2.02) & 2.41(2.59) & 2.48(2.18) & 2.33(2.78) & 2.35(2.70)\tabularnewline
				\cline{2-8}
				& \multirow{2}{*}{450} & 0.25 & 1.56(1.02) & 1.68(1.07) & 2.47(2.16) & 2.00(1.64) & 1.58(1.03)\tabularnewline
				\cline{3-8}
				& & 0.75 & 1.16(0.93) & 1.29(0.94) & 1.12(1.17) & 0.98(0.91) & 1.08(1.00)\tabularnewline
				\cline{1-8}
				\multirow{6}{*}{$\gamma_3$} & \multirow{2}{*}{50} & 0.25 & 7.91(5.62) & 7.85(5.78) & 12.81(7.72) & 11.55(6.88) & 10.14(4.98)\tabularnewline
				\cline{3-8}
				& & 0.75 & 6.85(3.74) & 6.53(3.45) & 11.51(8.13) & 8.61(4.61) & 7.94(3.66)\tabularnewline
				\cline{2-8}
				& \multirow{2}{*}{150} & 0.25 & 5.25(2.88) & 4.93(2.82) & 9.92(5.85) & 9.26(5.56) & 8.37(3.83)\tabularnewline
				\cline{3-8}
				& & 0.75 & 4.81(3.45) & 4.42(3.61) & 7.61(5.31) & 6.14(3.82) & 6.09(3.81)\tabularnewline
				\cline{2-8}
				& \multirow{2}{*}{450} & 0.25 & 3.50(2.13) & 3.22(1.88) & 8.46(4.52) & 8.14(4.22) & 8.06(4.63)\tabularnewline
				\cline{3-8}
				& & 0.75 & 2.18(0.98) & 1.79(0.73) & 4.52(3.02) & 2.88(1.23) & 3.35(1.50)\tabularnewline
				\cline{1-8}
				\multirow{6}{*}{$\gamma_4$} & \multirow{2}{*}{50} & 0.25 & 10.64(5.27) & 10.54(5.12) & 13.28(7.81) & 10.13(5.77) & 12.47(6.85)\tabularnewline
				\cline{3-8}
				& & 0.75 & 7.50(2.04) & 7.45(2.07) & 8.14(2.24) & 7.84(2.04) & 8.08(2.24)\tabularnewline
				\cline{2-8}
				& \multirow{2}{*}{150} & 0.25 & 8.57(2.15) & 8.51(2.19) & 8.97(3.72) & 8.40(4.88) & 9.19(5.88)\tabularnewline
				\cline{3-8}
				& & 0.75 & 6.93(1.36) & 6.89(1.41) & 7.17(1.33) & 7.23(1.23) & 6.67(1.30)\tabularnewline
				\cline{2-8}
				& \multirow{2}{*}{450} & 0.25 & 7.05(1.20) & 6.97(1.16) & 8.19(2.92) & 7.25(4.53) & 7.90(3.82)\tabularnewline
				\cline{3-8}
				& & 0.75 & 6.47(0.67) & 6.41(0.68) & 6.84(0.94) & 6.24(0.93) & 6.26(0.97)\tabularnewline
				\cline{1-8}
			\end{tabular}%
		}
	\end{center}
\end{table}

\begin{figure}
	\includegraphics[scale=0.5]{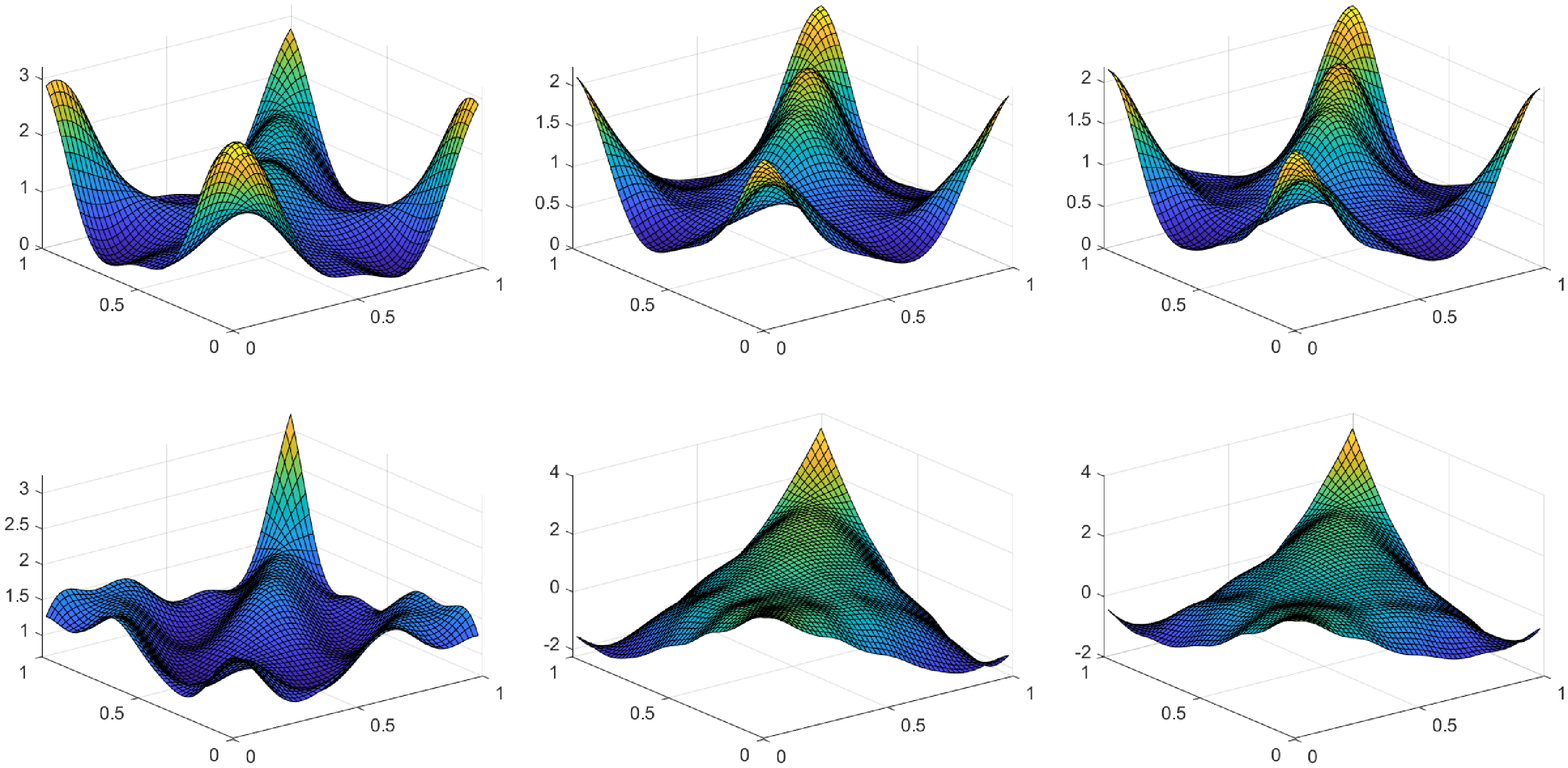}
	\caption{Estimated covariance functions corresponding to the median MISE for $\gamma=\gamma_1$; top left, true covariance function; top middle, LWZ estimate; top right, LWZ estimate without Fourier extension; bottom left, DHHK estimate; bottom middle, DP estimate; bottom right, ZC estimate.\label{fig:cov1-d25-n150}}
\end{figure}

\begin{figure}
	\includegraphics[scale=0.5]{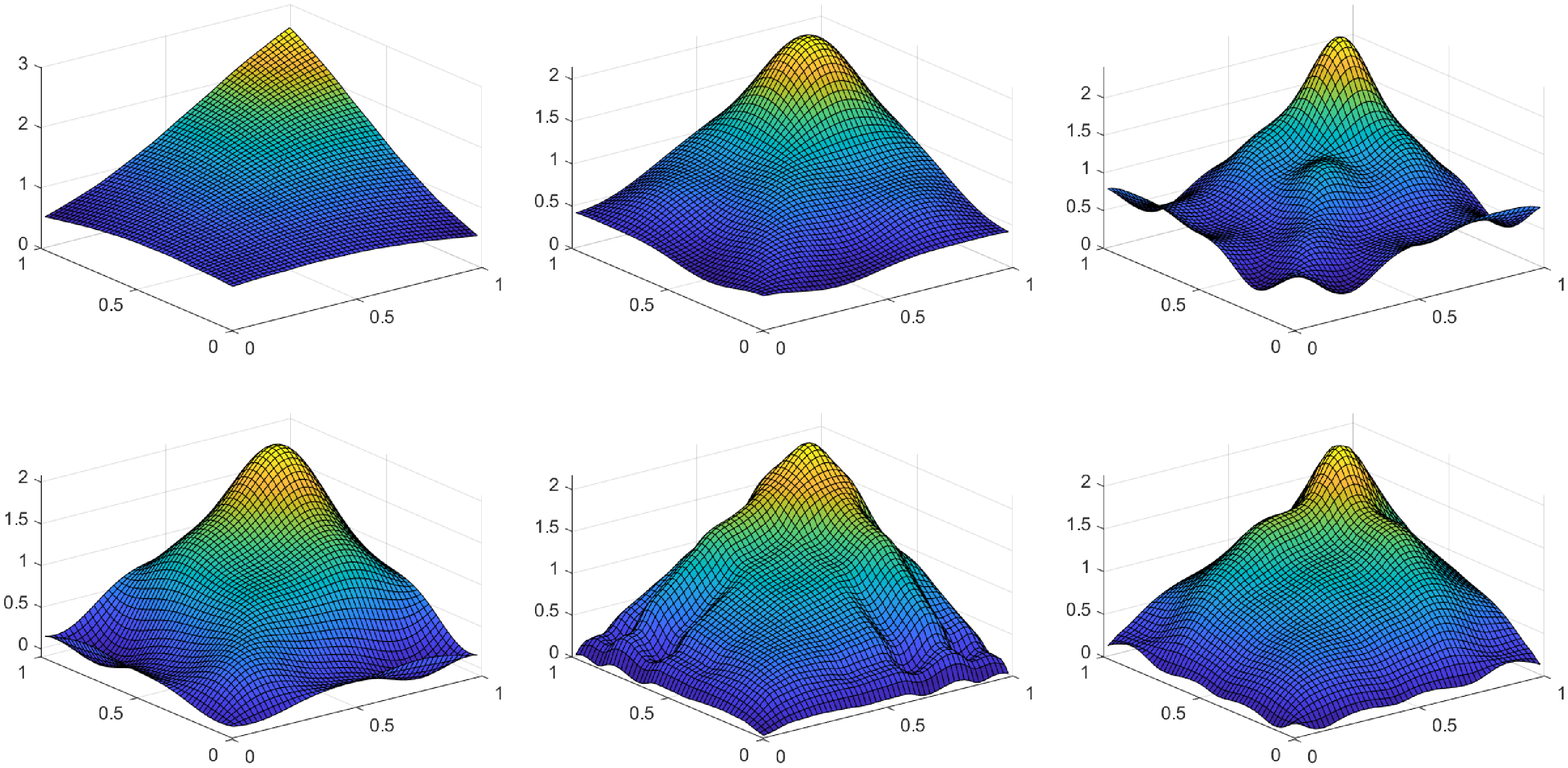}
	\caption{Estimated covariance functions corresponding to the median MISE for $\gamma=\gamma_2$; top left, true covariance function; top middle, LWZ estimate; top right, LWZ estimate without Fourier extension; bottom left, DHHK estimate; bottom middle, DP estimate; bottom right, ZC estimate.\label{fig:cov2-d25-n150}}
\end{figure}

\begin{figure}
	\includegraphics[scale=0.5]{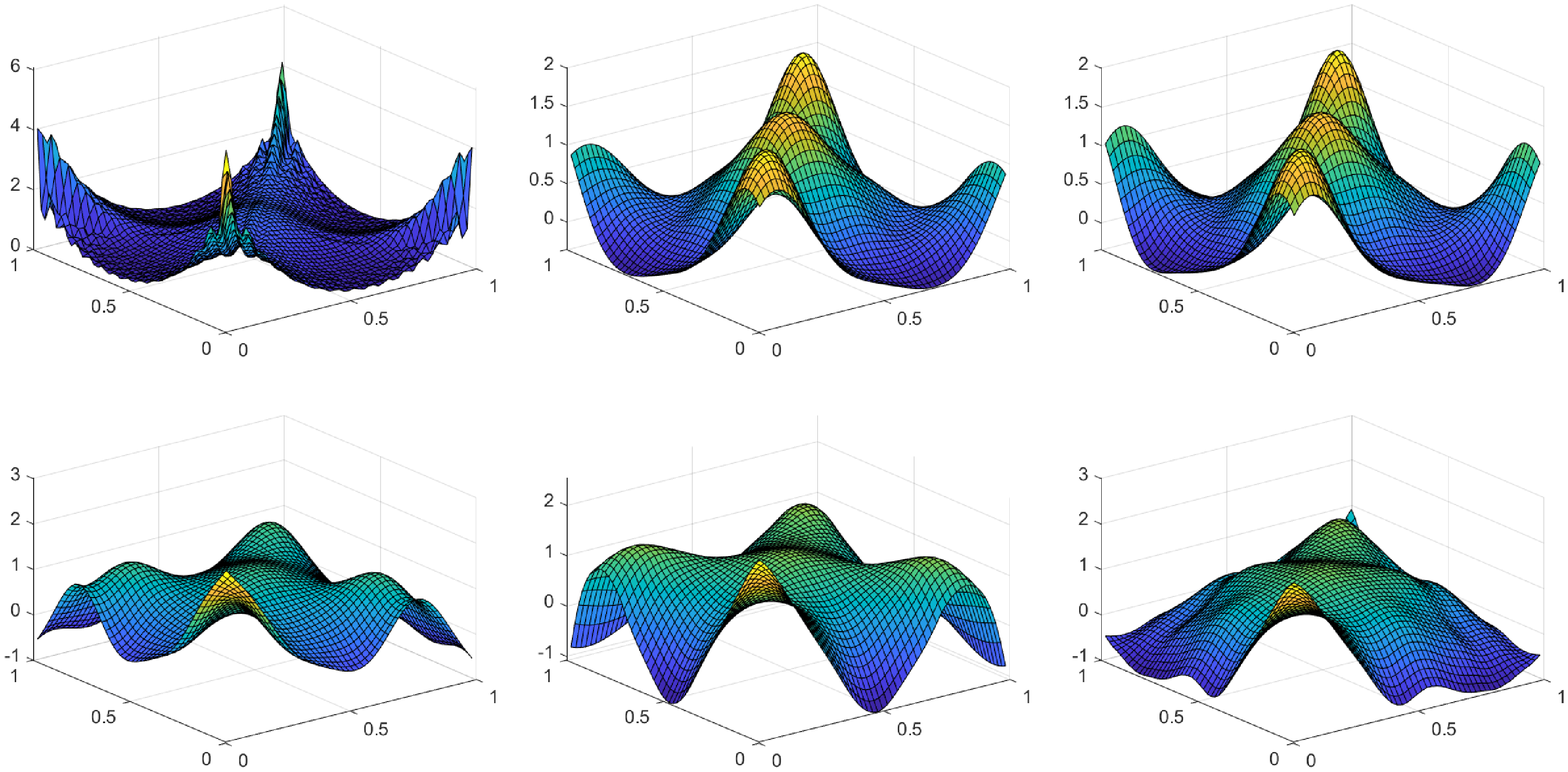}
	\caption{Estimated covariance functions corresponding to the median MISE for $\gamma=\gamma_3$; top left, true covariance function; top middle, LWZ estimate; top right, LWZ estimate without Fourier extension; bottom left, DHHK estimate; bottom middle, DP estimate; bottom right, ZC estimate.\label{fig:cov3-d25-n150}}
\end{figure}

\begin{figure}
	\includegraphics[scale=0.5]{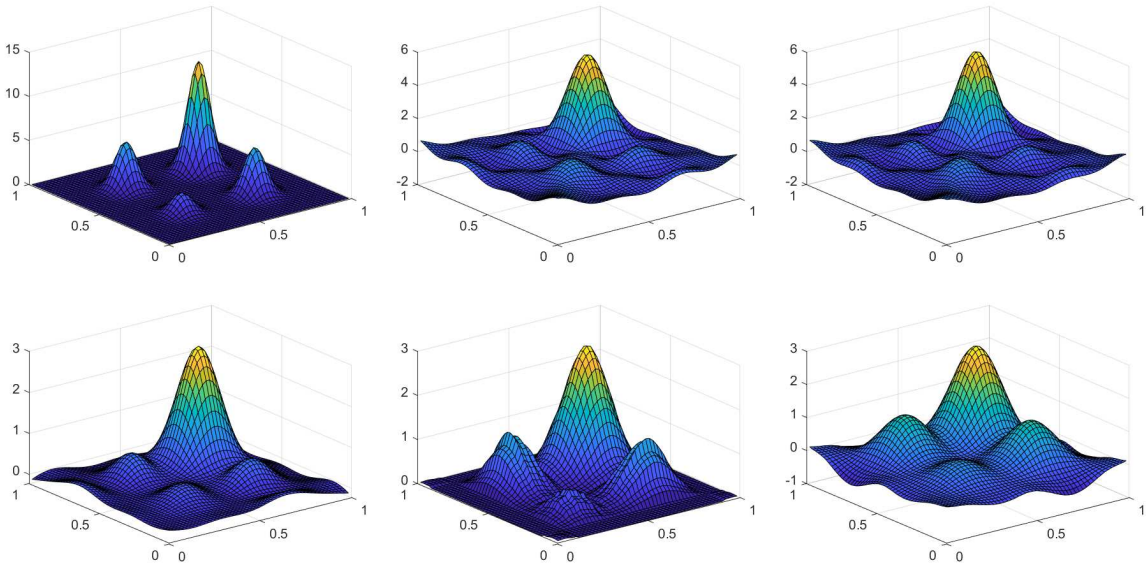}
	\caption{Estimated covariance functions corresponding to the median MISE for $\gamma=\gamma_4$; top left, true covariance function; top middle, LWZ estimate; top right, LWZ estimate without Fourier extension; bottom left, DHHK estimate; bottom middle, DP estimate; bottom right, ZC estimate.\label{fig:cov4-d25-n150}}
\end{figure}

\section{Applications}\label{sec:application}
\subsection{Spinal Bone Mineral Density}\label{subsec:sbmd}
In the study of \cite{BACHRACH1999}, 423 individuals with age ranging from 8 to 27, were examined for their longitudinal spinal bone mineral density. The bone density of each individual was irregularly recorded in four consecutive years, at most once for each year.  The data for each individual then lead to a functional snippet spanning at most 4 years.  In our study, individuals who have only one measurement are  excluded, since they do not carry information for the covariance structure. This results in  a total of 280 individuals who have at least two measurements and whose ages range from 8.8 to 26.2. 

 We are interested in the mean and covariance structure of the mineral density, the latter enabling us to derive the first few principal components. Figure \ref{fig:cov}(a) depicts  the empirical design of the covariance function,  underscoring the nature of these data as a collection of snippets: there is no data available to directly infer the off-diagonal region of the covariance structure. We also note that the design time points are  irregular. This feature renders techniques based on matrix completion less appropriate since they require a regular design for the measurement time points. In contrast, our method is able to  accommodate this irregularity. 
 
 The mineral density data and the estimated mean function are displayed in Figure \ref{fig:mean}(a) and \ref{fig:mean}(b), respectively. We observe that the mean density starts with a low level, rises rapidly before the age of 16 and then goes up relatively slowly to a peak at the age of around 20. This indicates  that the spinal bone mineral accumulates fast during adolescence, during which rapid physical growth and psychological changes occur, and then  remains at a stable high level in the early 20s. From Figure \ref{fig:mean}(a), we see that observations are relatively sparse from age 23 to age 26, especially near the boundary at age 26. Therefore, we suspect  that the upward trend around age 26 might be due to a boundary effect, in particular the marked rightmost point in Figure \ref{fig:mean}(a). To check this, we refit the data with this point excluded. The refitted mean curve is in  Figure \ref{fig:mean}(c), and indeed  the upward trend disappears. 
 
  The estimated covariance surface, after removing the rightmost point, is shown in Figure \ref{fig:cov}(b), which suggests larger variability of the data around the age of 17. It also indicates that the correlation of the longitudinal mineral density at different ages decays drastically as ages become more distant. \lin{Given the estimated covariance function in the entire domain, we are able to derive the first three principal components that are shown in Figure \ref{fig:bone-pc}. These principal components account for 69.6\%, 22.1\% and 3.6\% of the variance of the data, respectively. We see that the first three principal components can explain over 95\% of the total variance of the data. The first principal component, which explains nearly 70\% of the variation, shows that  the highest variation in the   bone density trajectories corresponds to  overall growth  that is  consistently either above or below  the mean  curve with the difference to the mean most prominent around age 15. Those with positive first principal component scores have bone densities consistently above the mean function with a surge at age 15, and vice versa for those with negative scores. The second principal component reflects the contrast of growth before and after age 17.5. Those with positive second scores have above average bone densities up to age 17.5 but then drop below the average afterwards.  
  The third principal component reflects the random fluctuation  of the bone growth.  }
 
 \begin{figure*}[t!]
    \centering
    \begin{subfigure}[t]{0.32\textwidth}
        \centering
        \includegraphics[scale=0.5]{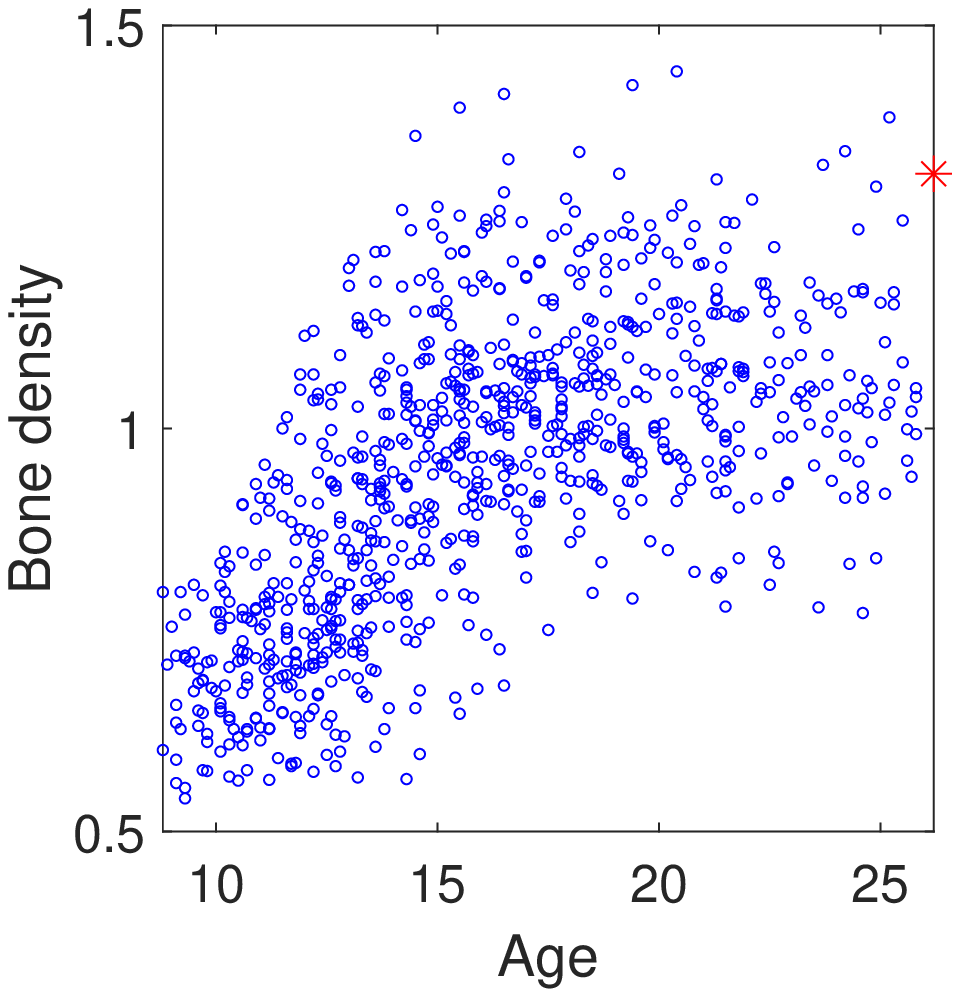}
        \caption{}
    \end{subfigure}%
    ~ 
    \begin{subfigure}[t]{0.32\textwidth}
        \centering
        \includegraphics[scale=0.5]{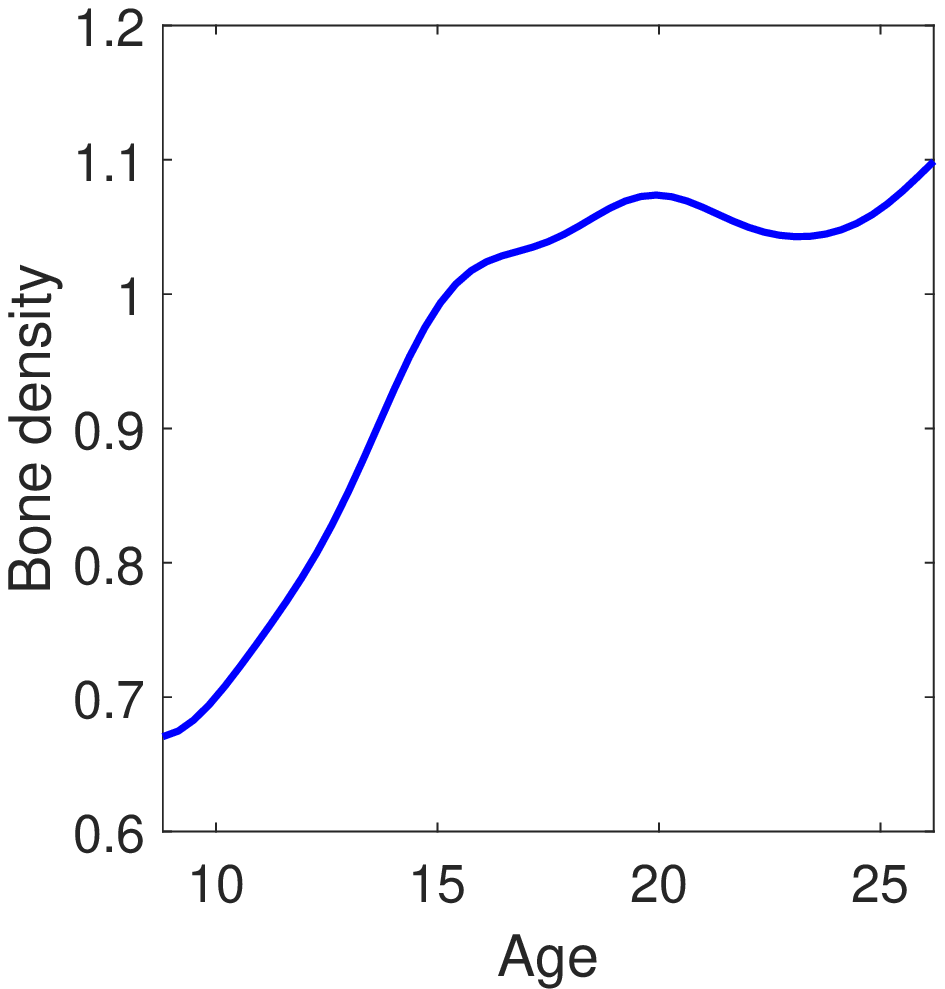}
        \caption{}
    \end{subfigure}
    ~
    \begin{subfigure}[t]{0.32\textwidth}
        \centering
        \includegraphics[scale=0.5]{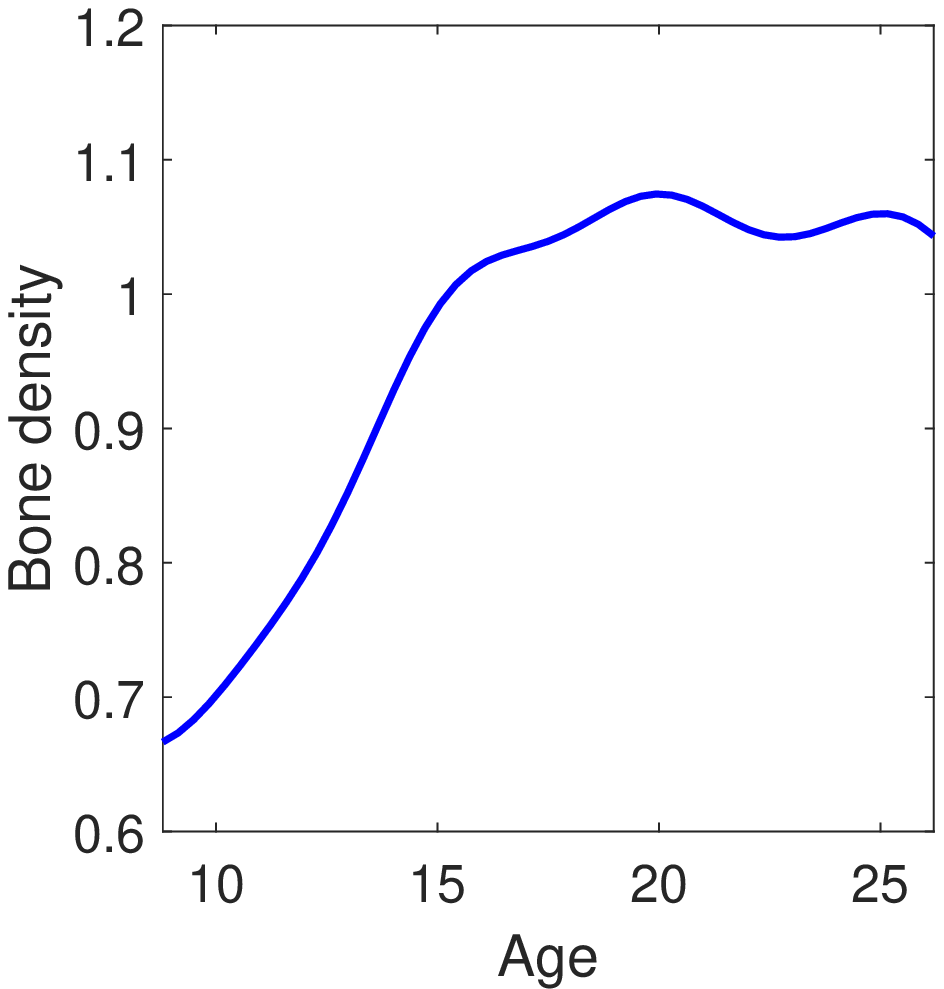}
        \caption{}
    \end{subfigure}
    \caption{(a) Spinal bone mineral density data. (b) Estimated mean function. (c) Estimated mean function  when the  rightmost point in the left panel is removed from the data.}\label{fig:mean}
\end{figure*}

\begin{figure*}[t!]
    \centering
    \begin{subfigure}[t]{0.47\textwidth}
        \centering
        \includegraphics[scale=0.6]{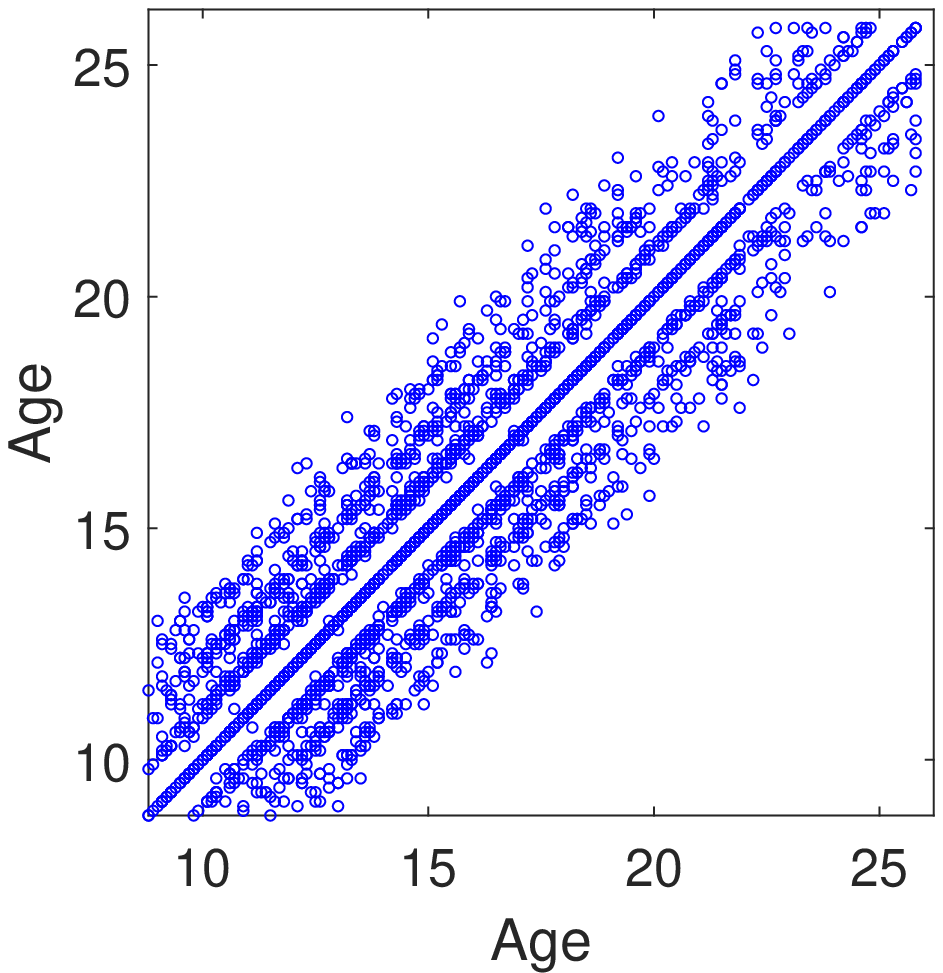}
        \caption{}
    \end{subfigure}%
    ~ 
    \begin{subfigure}[t]{0.53\textwidth}
        \centering
        \includegraphics[scale=0.6]{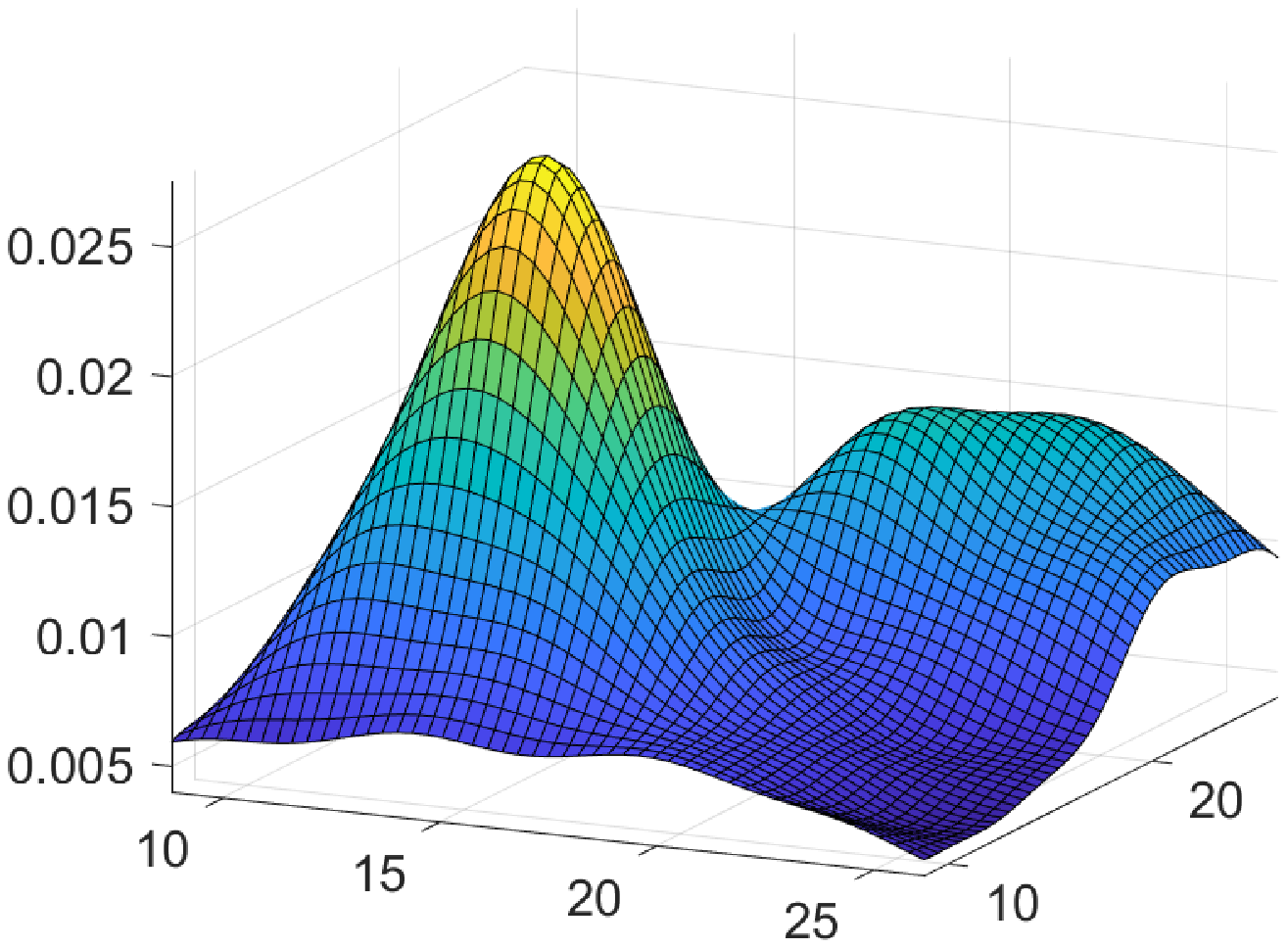}
        \caption{}
    \end{subfigure}
    \caption{(a) Empirical design plot of the covariance structure of  the spinal bone mineral density. (b) Estimated  covariance function of the  spinal bone mineral density, using the proposed method with Fourier basis and nonperiodic extension.}
\label{fig:cov}
\end{figure*}

 \begin{figure*}[t!]
    \centering
    \begin{subfigure}[t]{0.32\textwidth}
        \centering
        \includegraphics[scale=0.5]{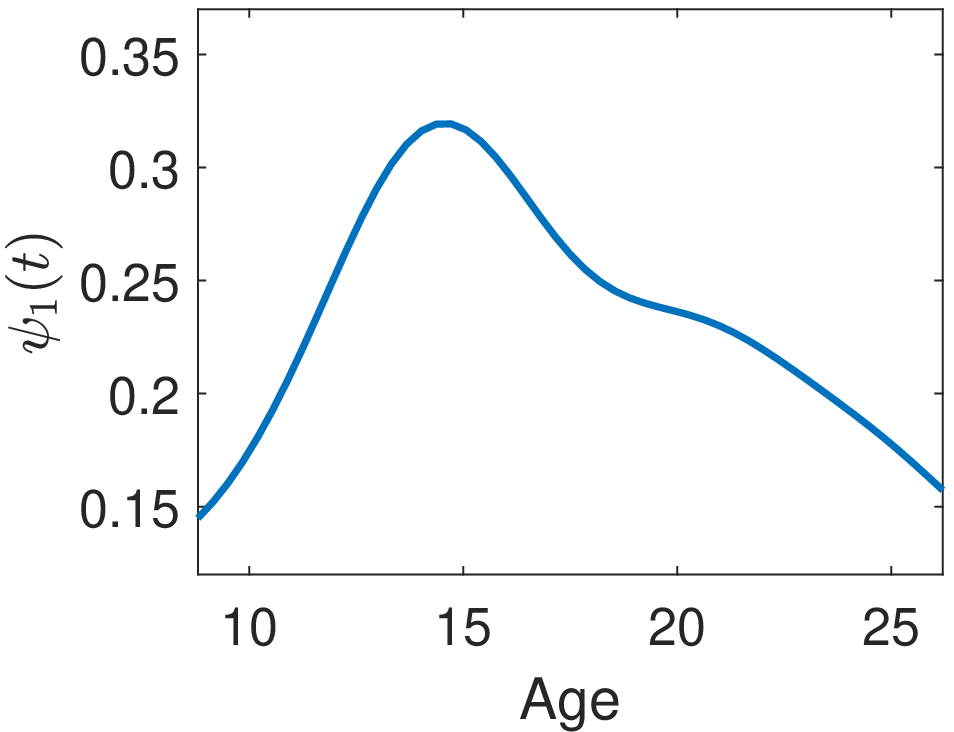}
        \caption{}
    \end{subfigure}%
    ~ 
    \begin{subfigure}[t]{0.32\textwidth}
        \centering
        \includegraphics[scale=0.5]{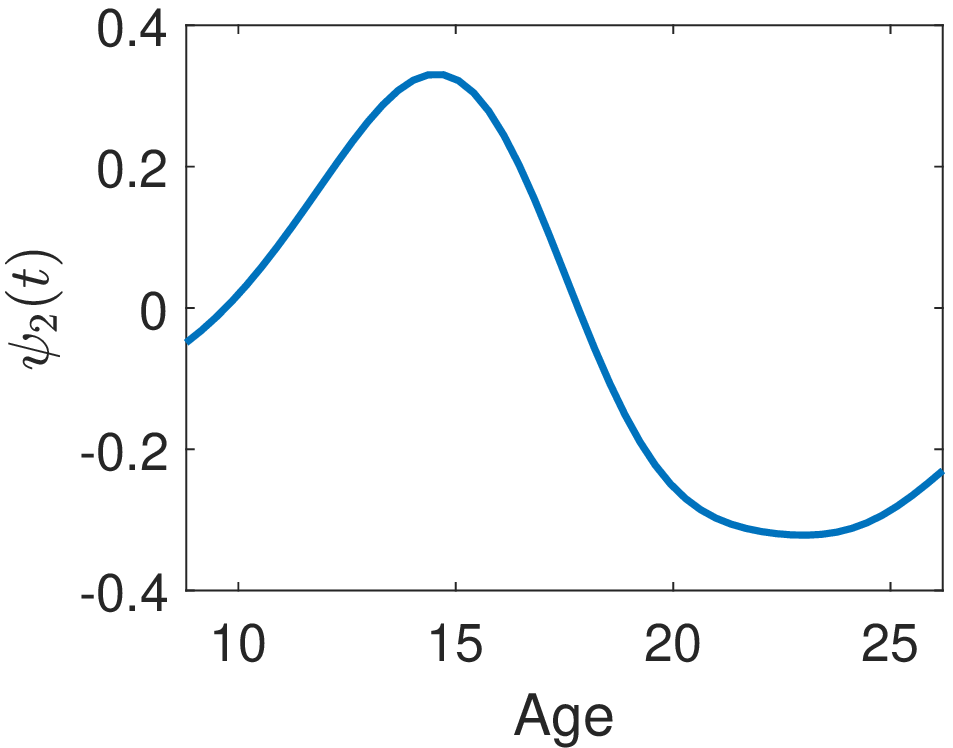}
        \caption{}
    \end{subfigure}
    ~
    \begin{subfigure}[t]{0.32\textwidth}
        \centering
        \includegraphics[scale=0.5]{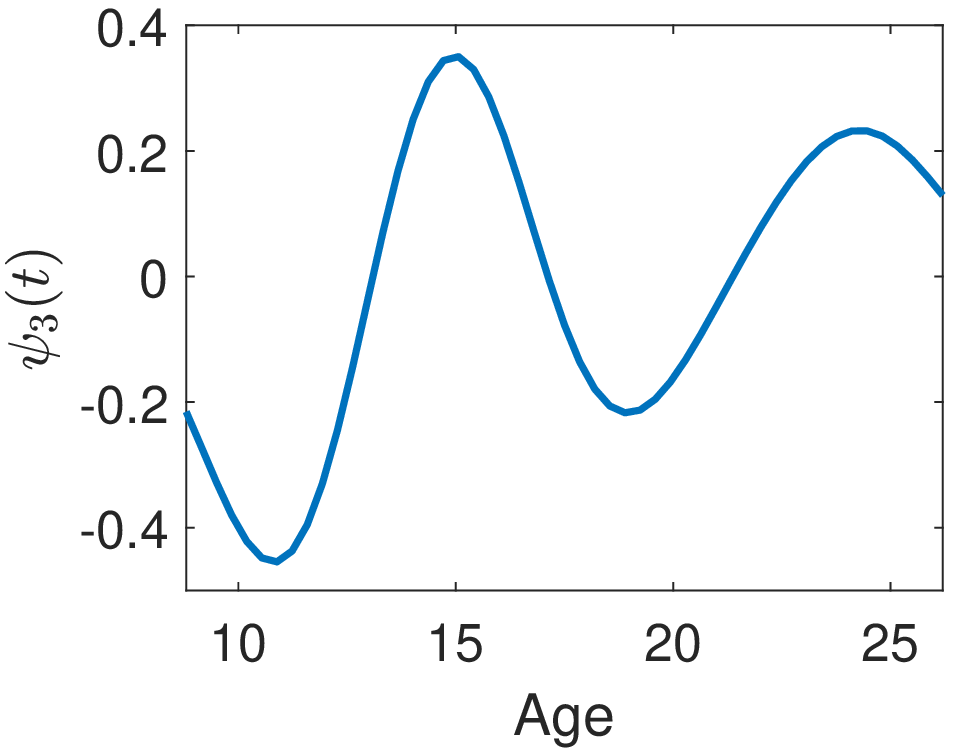}
        \caption{}
    \end{subfigure}
    \caption{The first three principal components of spinal bone density data}\label{fig:bone-pc}
\end{figure*}

\subsection{Systolic Blood Pressure}\label{subsec:sbp}
We apply our method to the study of age-associated changes in blood pressure of healthy men using an anonymous data from the Baltimore
Longitudinal Study of Aging (BLSA). In the study, 1590 healthy male volunteers (age 17--97) were scheduled to visit the Gerontology Research Center bi-annually. On each visit, systolic blood pressure (in mm Hg) was recorded along with other health related information.
Measurements taken within 2 years of death (49 visits in 44 men) are excluded to avoid potential bias due to diseases that might interfere with  blood pressure.
 The data have a sparse and irregular design as many visits were missed by participants or not on the schedule; see \cite{Pearson1997} for more details. Since there are very few data available beyond age 95, we restrict our focus to the age range 17--95. We also exclude subjects with only one visit. This results in a dataset of 1290 subjects, with approximate 7.5 visits per subject, shown in Figure \ref{fig:sbp-mean}(a).

From the design plot shown in Figure \ref{fig:sbp-cov}(a) for the covariance function, we see that there are no design points in the off-diagonal regions. Therefore, this is a dataset of functional snippets. Our estimated mean and covariance functions are shown in Figure \ref{fig:sbp-mean}(b) and Figure \ref{fig:sbp-cov}(b), respectively. We observe that the systolic blood pressure is relatively stable between 17 and 45, and then increases stiffly until the age 85. This is aligned with the discovery in \cite{Pearson1997}. The slight downward trend after 85 could be due to insufficient data around the boundary.  It could also be attributed, at least partly,  to the selection effect  \citep{Mueller1997, Wang1998}: people with lifespan over   85, termed  the oldest-old in aging research, are a selected group of subjects with favorable physical or genetic conditions.  

The estimated covariance function suggests an increasing trend of variability of systolic blood pressure along the age. Similar to the mean function, the slight downward trend after 85 could be due to the selection effect or lack of data. \lin{We also utilize the estimated covariance function in the entire domain to provide the first three principal components, shown in Figure \ref{fig:sbp-pc}. These principal components account for 72.3\%, 16.2\% and 6.07\% of the variance of the data, respectively.  The first principal component function has a similar shape as the mean function, reflecting that the largest variation of  systolic blood pressure is a vertical shift of the trajectory. This is a quite common phenomenon for functional data. The second principal component reflects  the contrast before and after age 68.  Those with positive second scores have higher  than average systolic blood pressure before age 68 but lower than average systolic blood pressure after age 68. The third principal components depicts the contrast before and after age 50. Those with positive third scores have higher than average systolic blood pressure before age 50, with a peak around age 40, but  lower than average blood pressure after age 50, probably due to change to a healthier life style.}  

\begin{figure*}[t!]
	\centering
	\begin{subfigure}[t]{0.47\textwidth}
		\centering
		\includegraphics[scale=0.6]{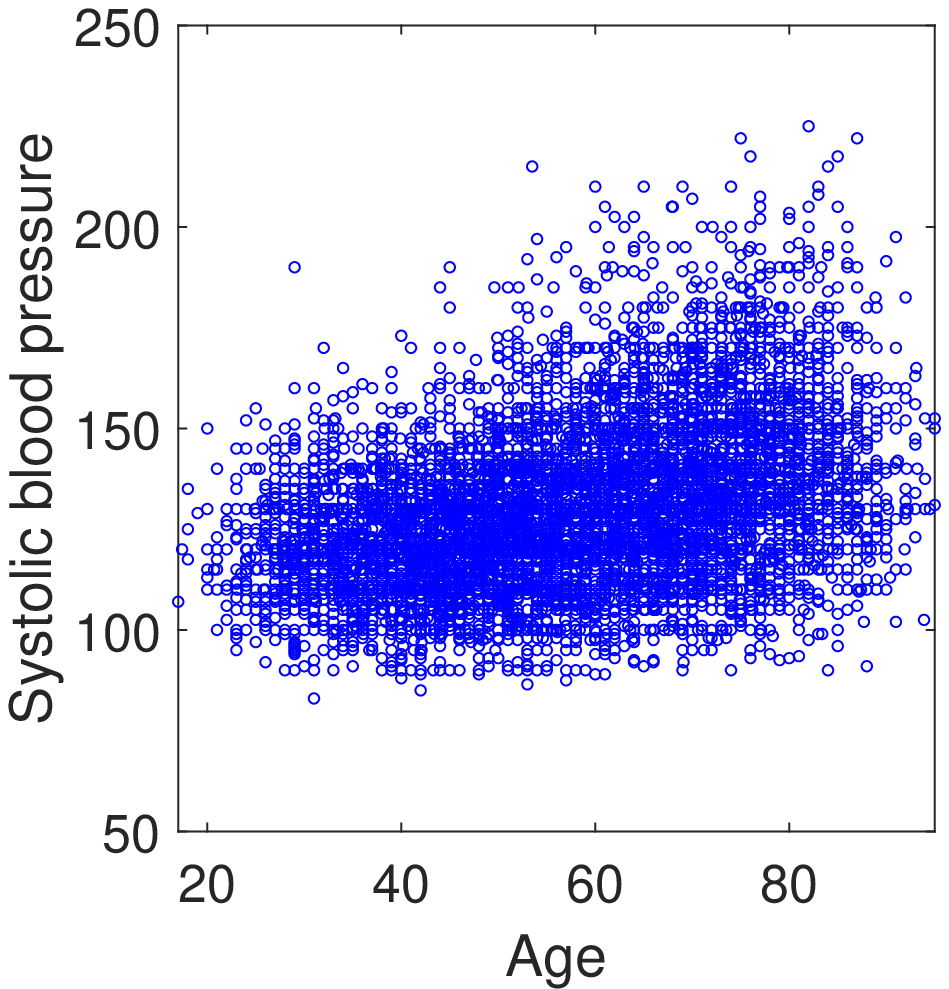}
		\caption{}
	\end{subfigure}%
	~ 
	\begin{subfigure}[t]{0.47\textwidth}
		\centering
		\includegraphics[scale=0.6]{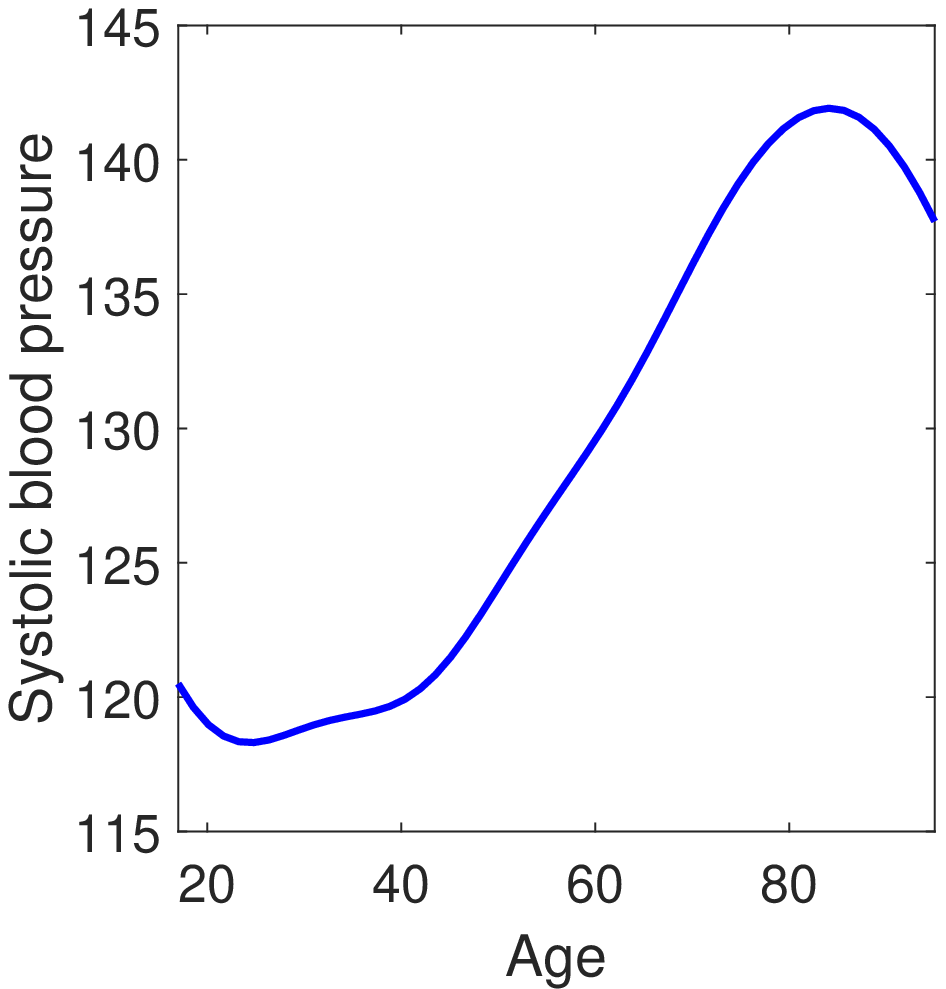}
		\caption{}
	\end{subfigure}

	\caption{(a)  Systolic blood pressure data. (b) Estimated mean function.}\label{fig:sbp-mean}
\end{figure*}

\begin{figure*}[t!]
	\centering
	\begin{subfigure}[t]{0.47\textwidth}
		\centering
		\includegraphics[scale=0.6]{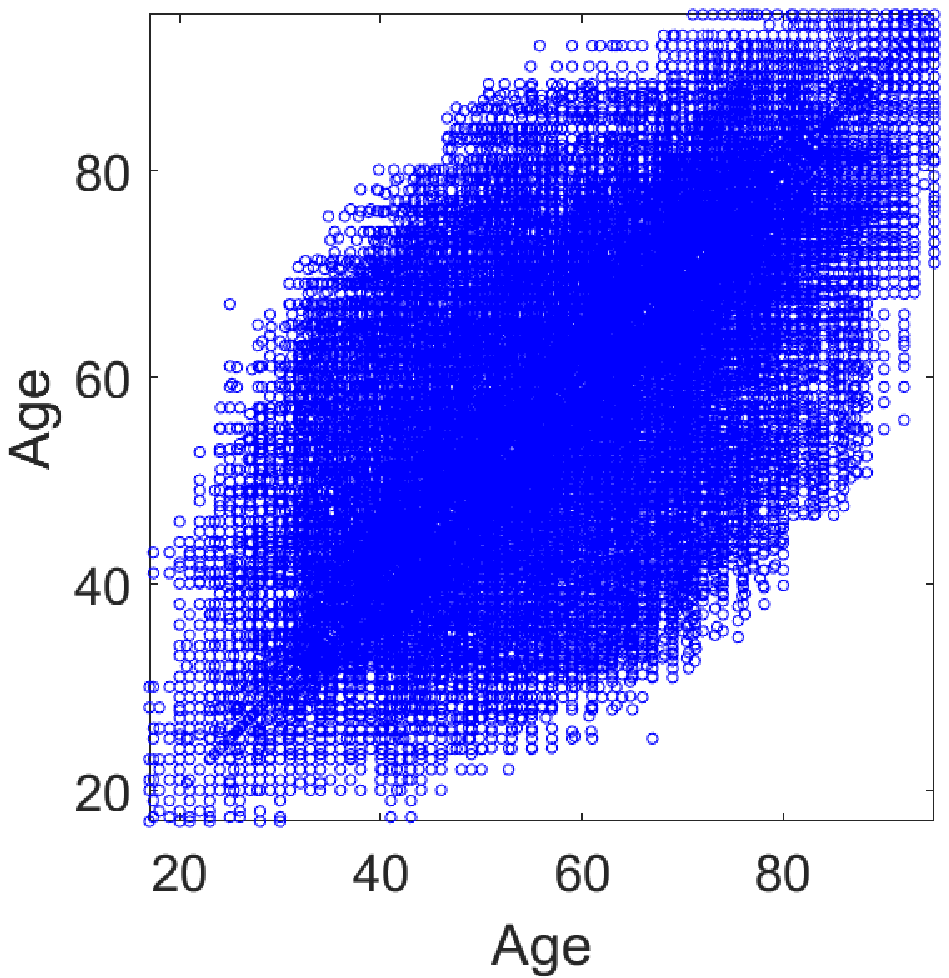}
		\caption{}
	\end{subfigure}%
	~ 
	\begin{subfigure}[t]{0.53\textwidth}
		\centering
		\includegraphics[scale=0.6]{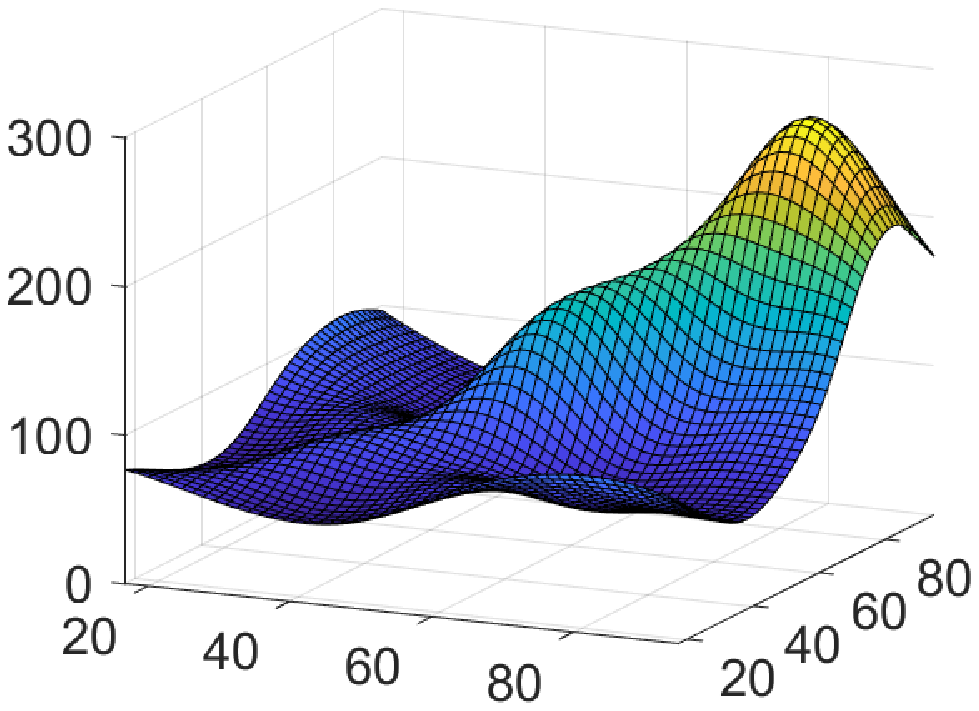}
		\caption{}
	\end{subfigure}
	\caption{(a) Empirical design plot of the covariance structure  of the systolic blood pressure. (b) Estimated  covariance function  of the  systolic blood pressure, using the proposed method with Fourier basis and nonperiodic extension.}
	\label{fig:sbp-cov}
\end{figure*}

\begin{figure*}[t!]
    \centering
    \begin{subfigure}[t]{0.32\textwidth}
        \centering
        \includegraphics[scale=0.5]{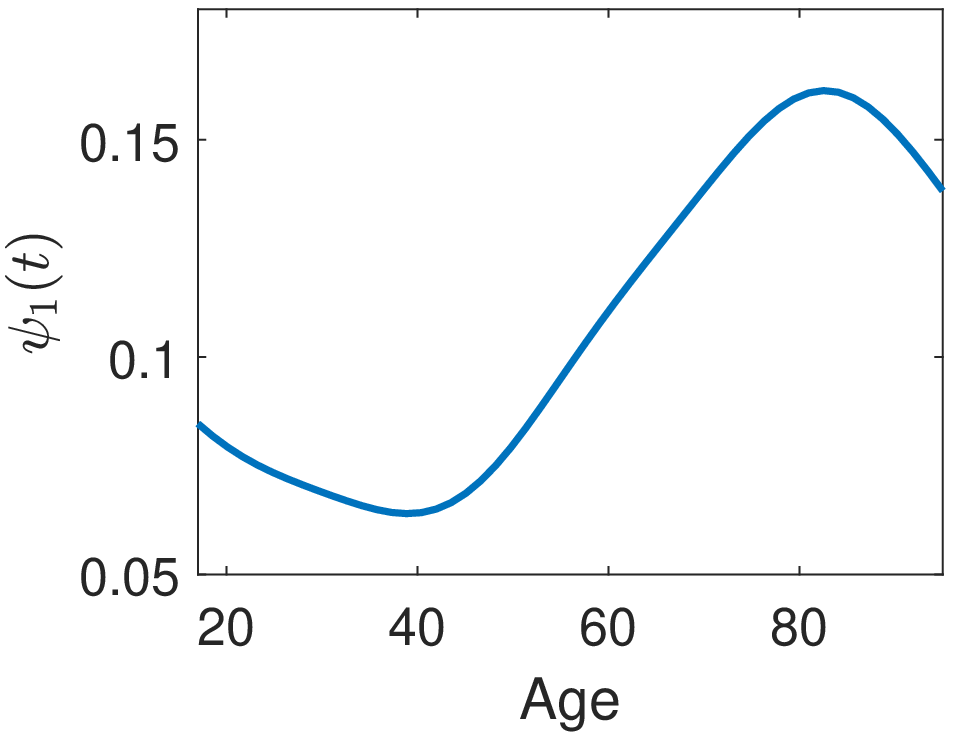}
        \caption{}
    \end{subfigure}%
    ~ 
    \begin{subfigure}[t]{0.32\textwidth}
        \centering
        \includegraphics[scale=0.5]{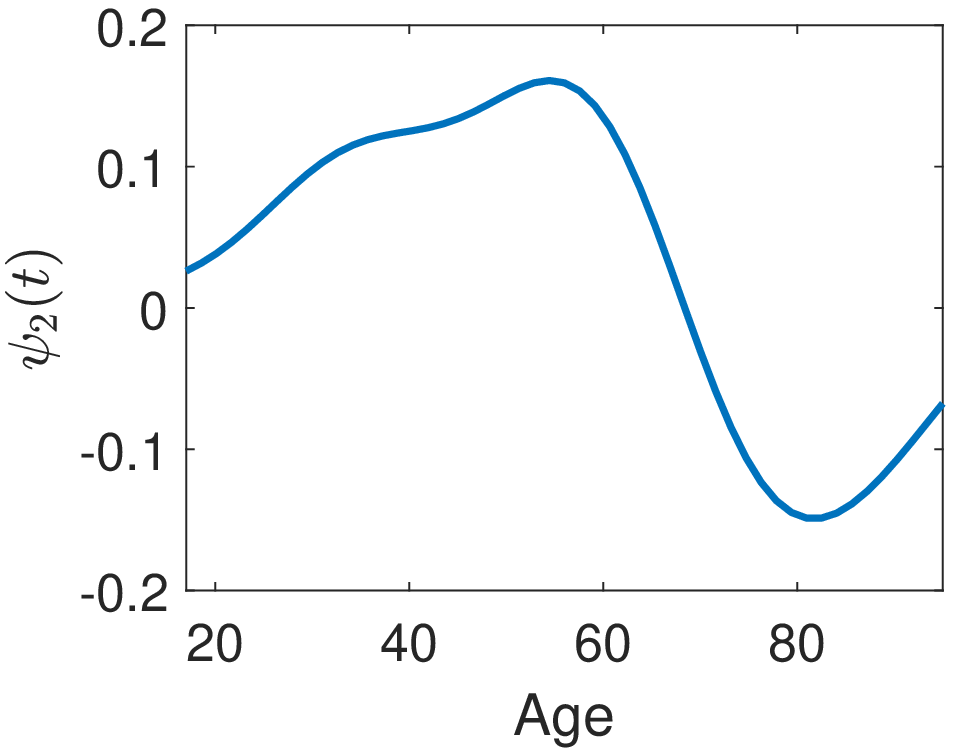}
        \caption{}
    \end{subfigure}
    ~
    \begin{subfigure}[t]{0.32\textwidth}
        \centering
        \includegraphics[scale=0.5]{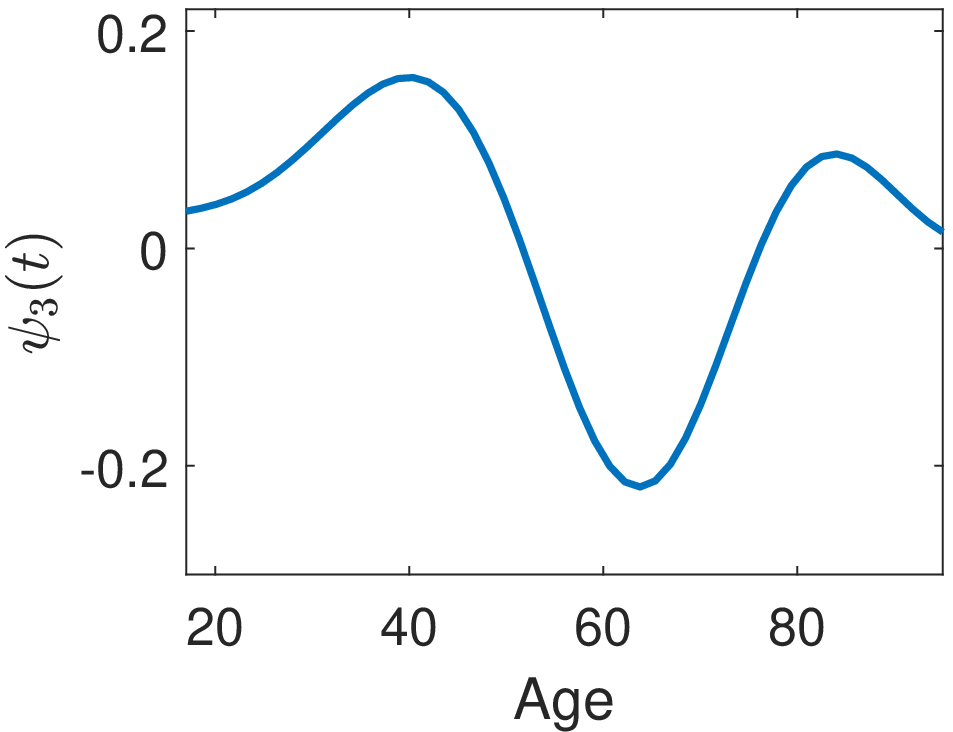}
        \caption{}
    \end{subfigure}
    \caption{The first three principal components of systolic blood pressure data}\label{fig:sbp-pc}
\end{figure*}

\section*{Acknowledgement}
We extend our sincere thanks to the editor and three anonymous reviewers for their detailed and constructive comments that help us substantially improve the paper.

\section*{Appendix A: Fourier Extension}
For a nonperiodic function $g$, its finite Fourier series expansion $g_N$ is observed to suffer from the so-called Gibbs phenomenon \citep{Zygmund2003} which refers to the drastic oscillatory overshoot in the region close to the two endpoints of the domain. 
 A remedy to this issue is to employ the technique of  Fourier extension mentioned in Example \ref{ex:1}.

The idea of Fourier extension is to approximate a nonperiodic function $g$ by Fourier basis functions defined on an \emph{extended domain} $\tdomain_{\zeta}=[-\zeta,1+\zeta]$ for some $\zeta>0$ that we refer to as  \emph{extension margin}.  The basis functions, are then defined by $\phi_1(t)=(1+2\zeta)^{-1/2}$, $\phi_{2k}(t)=\cos\{2k\pi t/(1+2\zeta)\}$ and $\phi_{2k-1}(t)=\sin\{2k\pi t/(1+2\zeta)\}$ for $k>1$. To elaborate, let $\mathscr G_N(\zeta)=\mathrm{span}\{\phi_1,\ldots,\phi_N\}$. The Fourier extension of $g$ within $\mathscr G_N(\zeta)$, denoted by $\tilde{g}_N$, is  defined by 
\begin{equation*}
\tilde{g}_N=\underset{h\in\mathscr G_N(\zeta)}{\arg\min}\,\|g-h\|_{\ltwo(0,1)},
\end{equation*}
where we emphasize  that the norm $\|\cdot\|_{\ltwo(0,1)}$ is for the domain $\tdomain$.   One can easily see that the Fourier extension of $g$ is not unique. However,  all such extensions have the same approximation quality for $g$ over the domain $\tdomain$ of interest. Intuitively, for a Fourier extension of $g$, even if there is a Gibbs phenomenon, for a suitable $\zeta$, the phenomenon is expected to occur only within the extended part of the domain, i.e., $[-\zeta,0]$ and $[1,1+\zeta]$.  Consequently, the nonperiodic function $g$ can then be well approximated on the domain $[0,1]$. Indeed, the speed of the convergence of a Fourier extension of $g$ in the domain $[0,1]$ adapts to the smoothness of $g$ \citep{Adcock2014}. For example, $\tilde{g}_N$ converges to $g$ at the rate of $N^{-r}$ ($c^{-N}$ for some $c>1$, respectively) when $g$ is $r$ times differentiable (analytic, respectively). 

The above discussion can be straightforwardly extended to the two-dimensional case. Let $\mathscr G_N^2(\zeta)=\mathrm{span}\{\phi_k\otimes\phi_l:1\leq k,l\leq N\}$, where $\phi_k\otimes\phi_l$ represents the function $\phi_k(s)\phi_l(t)$ defined on the two-dimensional square $[-\zeta,1+\zeta]^2$. For a function $\gamma$ defined on $[0,1]^2$, its Fourier extension $\tilde{\gamma}$ within $\mathscr G_N^2(\zeta)$ is given by 
\begin{equation*}
\tilde{\gamma}_N=\underset{h\in\mathscr G_N^2(\zeta)}{\arg\min}\,\|\gamma-h\|_{\ltwo([0,1]^2)},
\end{equation*}
where $\ltwo([0,1]^2)$ denotes the space of squared integrable functions defined on $[0,1]^2$ with the norm $\|\gamma-h\|_{\ltwo([0,1]^2)}=\{\int_0^1\int_0^1 |\gamma(s,t)|^2\diffop s\diffop t\}^{1/2}$.

\section*{Appendix B: Geometric Newton Method}
Let $\mathbf{C}=\mathbf{L}\mathbf{L}^{\tp}$, where $\mathbf{L}$
is a lower triangular matrix whose diagonal elements are all positive.
This is the so-called Cholesky decomposition of $\mathbf{C}$ and
it is unique (since we require the diagonal elements of $\mathbf{L}$
to be positive). The objective function in \eqref{Mrep} can be written
as 
\begin{align*}
Q(\mathbf{L})= & \sum_{i=1}^{n}\tr[(\mathbf{R}_{i}-\mathbf{B}_{i}\mathbf{L}\mathbf{L}^{\tp}\mathbf{B}_{i}^{\tp})(\mathbf{R}_{i}-\mathbf{B}_{i}\mathbf{L}\mathbf{L}^{\tp}\mathbf{B}_{i}^{\tp})^{\tp}]\\
& -\sum_{i=1}^{n}\sum_{j=1}^{N_{i}}(\Gamma_{ijj}-\mathbf{b}_{ij}\mathbf{L}\mathbf{L}^{\tp}\mathbf{b}_{ij}^{\tp})^{2}\\
& +\lambda\tr(\mathbf{\mathbf{L}\mathbf{L}^{\tp}U\mathbf{L}\mathbf{L}^{\tp}W}+\mathbf{\mathbf{L}\mathbf{L}^{\tp}V\mathbf{L}\mathbf{L}^{\tp}V}),
\end{align*}
where $\mathbf{R}_{i}$ is a matrix whose $(j,k)$-entry is $\Gamma_{ijk}$,
$\mathbf{B}_{i}$ is a matrix whose $(j,k)$-entry is $\phi_{j}(T_{ik})$
for $j=1,\ldots,p$ and $k=1,\ldots,m_{i}$, and $\mathbf{b}_{ij}$
is a row vector representing the $j$th row of $\mathbf{B}_{i}$. Instead of a constraint optimization problem in the space of all $p\times p$ lower triangular matrices, we view the minimization of $Q$ as an unconstrained optimization problem in the manifold of $p\times p$ lower triangular matrices whose diagonal elements are all positive, and then adopt Newton's method on manifolds, as follows.

Informally speaking, a $d$-dimensional manifold $\mathcal M$ can be viewed as a subset of $\real^D$ for $d\leq D$, such that, for every $x\in\mathcal{M}$, there is a $d$-dimensional hyperplane $\mathbb{T}_x$ (of $\real^D$) that is tangential to $\mathcal{M}$ at $x$. If one endows the hyperplane $\mathbb{T}_x$ with an inner product, and the inner product (depending on $x$) varies smoothly with respect to $x$, then the manifold is called a Riemannian manifold. For a Riemannian manifold, for each $x$, there is an map $\mathrm{Exp}_x:\mathbb T_x\rightarrow\mathcal M$, called the Riemannian exponential map at $x$, such that for $u\in\mathbb{T}_x\subset\real^D$, $\mathrm{Exp}(u)=\psi(1)$ for a geodesic $\psi:[-1,1]\rightarrow\mathcal M$ passing through $x$ and $\psi^\prime(0)=u\in\real^D$. For an in-depth introduction to manifolds, we recommend the textbook by \cite{Lang1995}. For the manifold $\mathcal{M}=\{\mathbf L\in\real^{p\times p}:\,\mathbf{L}$ is lower triangular and has positive diagonal elements$\}$, \cite{Lin2019+a} introduced a way to define smoothly varying inner products among the hyperplanes $\mathbb{T}_{\mathbf X}$ for $\mathbf X\in\mathcal M$, which turn $\mathcal{M}$ into a Riemannian manifold with Riemannian exponential map defined by $\mathrm{Exp}_{\mathbf{X}}\mathbf{S}=(\mathbf{X}-\mathbb{D}(\mathbf{X}))+(\mathbf{S}-\mathbb{D}(\mathbf{S}))+\mathbb{D}(\mathbf{X})\mathrm{expm}\{\mathbb{D}(\mathbf{S})\mathbb{D}(\mathbf{X})^{-1}\}$ for $\mathbf X\in\mathcal M$ and $\mathbf S\in\mathbb T_{\mathbf X}$, where $\mathbb D(\mathbf X)$ denotes the diagonal matrix formed by the diagonal part of $\mathbf X$, and $\mathrm{expm}(\mathbf A)=\sum_{k=0}^\infty \mathbf{A}^k /k!$, defined for any square matrix $\mathbf A$, denotes the matrix exponential function.

\lin{The geometric Newton method is the generalization of classic Newton method to Riemannian manifolds to compute a zero of a function $F(x)\in\mathbb{T}_x$. It can be used to solve our minimization problem by observing that if $Q(\mathbf L)$ is (locally) minimized at $\mathbf{L}^\ast$ then $\partial Q/\partial \mathbf{L}$ vanishes at $\mathbf L^\ast$. In this case, $F(\mathbf L)=\partial Q/\partial \mathbf{L}$. Starting with an initial point $\mathbf L_0\in\mathcal M$, for $k=0,1,\ldots$, we first solve the Newton equation $$ \mathcal H_{\mathbf L_k}\eta_k=-F(\mathbf{L}_k)$$
for the unknown $\eta_k\in\mathbb T_{\mathbf L_k}$, 
where $\mathcal H_{\mathbf L}:\mathbb T_{\mathbf L}\rightarrow\mathbb T_{\mathbf L}$ denotes the gradient of $F$ (or equivalently, the Hessian of $Q$), and then set $\mathbf{L}_{k+1}=\mathrm{Exp}_{\mathbf L_k}(\eta_k)$. The iteration is terminated if 1) $k\geq K$ for some specified integer $K>0$, 2) $\|F(\mathbf{L}_k)\|_2\leq \epsilon$ for a tolerance level $\epsilon>0$, or 3) $\gamma_{{\vec L}_{k+1}{\vec L}_{k+1}^\tp}\not\in\mathcal C$. One also needs to select an initial point $\mathbf{L}_0$, which is determined in the following way. We first minimize $Q$ with respect to $\mathbf C$ without considering the positive definiteness and the constraint $\gamma_{\vec C}\in\mathcal C$. In this case, $Q$ is a quadratic form of $\vec C$ and a closed-form solution is available. Denote the minimizer by $\check{\vec C}$. If $\check{\vec C}$ is not positive definite,  we obtain its eigendecomposition $\check{\vec C}=\check{\vec P}\check{\Lambda}\check{\vec P}^\tp$ for an orthogonal matrix $\check{\vec P}$ and a diagonal matrix $\check\Lambda$, and then set $\check{\vec C}=\check{\vec P}\mathring{\Lambda}\check{\vec P}^\tp$, where $\mathring{\Lambda}$ is obtained from $\check\Lambda$ by setting its nonpositive diagonal elements to $0.01\times$largest singular value of $\check{\vec C}$.}

It remains to compute the gradient $F$ and the Hessian $\mathcal H$ of the objective function $Q$. Straightforward computation shows that 
\begin{align*}
F(\mathbf L)= & -4\sum_{i=1}^{n}\mathbf{B}_{i}^{\tp}\mathbf{R}_{i}\mathbf{B}_{i}\mathbf{L}\\
& +4\sum_{i=1}^{n}\mathbf{B}_{i}^{\tp}\mathbf{B}_{i}\mathbf{L}\mathbf{L}^{\tp}\mathbf{B}_{i}^{\tp}\mathbf{B}_{i}\mathbf{L}\\
& -4\sum_{i=1}^{n}\sum_{j=1}^{N_{i}}\mathbf{b}_{ij}^{\tp}\mathbf{b}_{ij}\mathbf{L}\mathbf{L}^{\tp}\mathbf{b}_{ij}^{\tp}\mathbf{b}_{ij}\mathbf{L}\\
& +4\sum_{i=1}^{n}\sum_{j=1}^{N_{i}}\Gamma_{ijj}\mathbf{b}_{ij}^{\tp}\mathbf{b}_{ij}\mathbf{L}\\
& +2\lambda\{\mathbf{U}\mathbf{L}\mathbf{L}^{\tp}\mathbf{W}\mathbf{\mathbf{L}}+\mathbf{W}\mathbf{L}\mathbf{L}^{\tp}\mathbf{U}\mathbf{L}\}\\
& +4\lambda\mathbf{V}\mathbf{L}\mathbf{L}^{\tp}\mathbf{V}\mathbf{\mathbf{L}}.
\end{align*}
	
	To compute the Hessian of $Q$, we first observe that, according to Proposition 5.5.4 of \cite{AbsMahSep2008}, 
\[
\mathcal H_{\mathbf{L}}=\mathrm{Hess}(Q\circ\mathrm{Exp}_{\mathbf{L}})(0),
\]
where $\circ$ denotes the composition of functions. 
Let $\mathbf{X}=\mathrm{Exp}_{\mathbf{L}}\mathbf{S}$ and $\Psi(\mathbf{S})=(Q\circ\mathrm{Exp}_{\mathbf{L}})(\mathbf{S})$.
Then
\begin{align*}
\mathrm{grad}\Psi= & -4\sum_{i=1}^{n}(\mathbf{B}_{i}^{\tp}\mathbf{R}_{i}\mathbf{B}_{i}\mathbf{X})\odot\mathbf{Z}\\
& +4\sum_{i=1}^{n}(\mathbf{B}_{i}^{\tp}\mathbf{B}_{i}\mathbf{X}\mathbf{X}^{\tp}\mathbf{B}_{i}^{\tp}\mathbf{B}_{i}\mathbf{X})\odot\mathbf{Z}\\
& -4\sum_{i=1}^{n}\sum_{j=1}^{N_{i}}(\mathbf{b}_{ij}^{\tp}\mathbf{b}_{ij}\mathbf{X}\mathbf{X}^{\tp}\mathbf{b}_{ij}^{\tp}\mathbf{b}_{ij}\mathbf{X})\odot\mathbf{Z}\\
& +4\sum_{i=1}^{n}\sum_{j=1}^{N_{i}}\Gamma_{ijj}(\mathbf{b}_{ij}^{\tp}\mathbf{b}_{ij}\mathbf{X})\odot\mathbf{Z}\\
& +2\lambda\{\mathbf{U}\mathbf{X}\mathbf{X}^{\tp}\mathbf{W}\mathbf{\mathbf{X}}+\mathbf{W}\mathbf{X}\mathbf{X}^{\tp}\mathbf{U}\mathbf{X}\}\odot\mathbf{Z}\\
& +4\lambda(\mathbf{V}\mathbf{X}\mathbf{X}^{\tp}\mathbf{V}\mathbf{\mathbf{X}})\odot\mathbf{Z},
\end{align*}
where $\odot$ denotes matrix Hadamard product, and $\mathbf{Z}=\{z_{jk}\}$
is a $p\times p$ lower triangular matrix such that $z_{jk}=1$ if
$k<j$ and $z_{jj}=e^{s_{jj}/L_{jj}}$ with $L_{jj}$ being the $(j,j)$-entry of $\mathbf L$. To compute the Hessian of
$Q\circ\mathrm{Exp}_{\mathbf{L}}$ at $\mathbf{S}$, we proceed to compute the gradient of $\mathrm{grad}\Psi$. 
We first  observe that, the terms in $\mathrm{grad} \Psi$
can be divided into two types: one is of the form $\mathbf{H}\mathbf{X}\odot\mathbf{Z}$,
and the other is of the form $\mathbf{H}\mathbf{X}\mathbf{X}^{\tp}\mathbf{G}\mathbf{\mathbf{X}}\odot\mathbf{Z}$,
where $\mathbf{H}$ and $\mathbf{G}$ are some matrices not containing
$\mathbf{X}$ or $\mathbf{S}$. Below we address them separately.

For the first type, let
\begin{align*}
\Xi_{jk}\equiv(\mathbf{H}\mathbf{\mathbf{X}})_{jk} & =\sum_{\ell}h_{j\ell}x_{\ell k},
\end{align*}
where $h_{j\ell}$ and $x_{\ell k}$ are the $(j,\ell)$- and $(\ell,k)$-entry of $\mathbf{H}$ and $\mathbf{X}$, respectively. Below we compute
$\frac{\partial\Xi_{jk}}{\partial s_{\tau\nu}}$ for $1\leq k\leq j\leq p$
and $1\leq\nu\leq\tau\leq p$. When $k\neq\nu$, $L_{\ell k}$ does
not contain $s_{\tau\nu}$, and thus
\[
\frac{\partial\Xi_{jk}}{\partial s_{\tau\nu}}=\frac{\partial}{\partial s_{\tau\nu}}\sum_{\ell}h_{j\ell}x_{\ell k}=0.
\]
When $k=\nu$, for $\nu=\tau$,
\[
\frac{\partial\Xi_{jk}}{\partial s_{\tau\nu}}=\frac{\partial}{\partial s_{\tau\nu}}\sum_{\ell}h_{j\ell}x_{\ell k}=e^{s_{\tau\nu}/L_{\tau\nu}}h_{j\tau},
\]
and for $\nu<\tau$, 
\[
\frac{\partial\Xi_{jk}}{\partial s_{\tau\nu}}=\frac{\partial}{\partial s_{\tau\nu}}\sum_{\ell}h_{j\ell}x_{\ell k}=h_{j\tau}.
\]

To deal the second type, let
\begin{align*}
\Xi_{jk}\equiv(\mathbf{H}\mathbf{X}\mathbf{X}^{\tp}\mathbf{G}\mathbf{\mathbf{X}})_{jk} & =\sum_{\omega}\sum_{\ell}\sum_{a}\sum_{b}h_{j\omega}x_{\omega\ell}x_{a\ell}g_{ab}x_{bk}.
\end{align*}
Now we compute $\frac{\partial\Xi_{jk}}{\partial s_{\tau\nu}}$
for $1\leq k\leq j\leq p$ and $1\leq\nu\leq\tau\leq p$. When $k\neq\nu$,
for $\nu=\tau$,

\[
\frac{\partial\Xi_{jk}}{\partial s_{\tau\nu}}=e^{s_{\tau\nu}/L_{\tau\nu}}\sum_{a}\sum_{b}(h_{j\tau}x_{a\nu}g_{ab}x_{bk})+e^{s_{\tau\nu}/L_{\tau\nu}}\sum_{\omega}\sum_{b}h_{j\omega}x_{\omega\nu}g_{\tau b}x_{bk},
\]
and for $\nu<\tau$, 
\begin{align*}
\frac{\partial\Xi_{jk}}{\partial s_{\tau\nu}} &  =\sum_{a}\sum_{b}h_{j\tau}x_{a\nu}g_{ab}x_{bk}+\sum_{\omega}\sum_{b}h_{j\omega}x_{\omega\nu}g_{\tau b}x_{bk}\\
& =h_{j\tau}\mathbf{X}_{\cdot\nu}^{\tp}\mathbf{G}\mathbf{X}_{\cdot k}+\mathbf{H}_{j\cdot}\mathbf{X}_{\cdot\nu}\mathbf{G}_{\tau\cdot}\mathbf{X}_{\cdot k},
\end{align*}
where the notation $\mathbf{X}_{\cdot k}$ denotes the $k$th column
of the matrix $\mathbf{X}$, and $\mathbf{G}_{\tau\cdot}$ denotes
the $\tau$th row of $\mathbf{G}$.

When $k=\nu$, for $\tau=\nu$,
\begin{align*}
\frac{\partial\Xi_{jk}}{\partial s_{\tau\nu}} &  =e^{s_{\tau\nu}/L_{\tau\nu}}\left(\sum_{a}\sum_{b}h_{j\tau}x_{a\nu}g_{ab}x_{bk}+\sum_{\ell}\sum_{b}h_{j\ell}x_{\ell\nu}g_{\tau b}x_{bk}+\sum_{\omega}\sum_{\ell}\sum_{a}h_{j\omega}x_{\omega\ell}x_{a\ell}g_{a\tau}\right)\\
& =e^{s_{\tau\nu}/L_{\tau\nu}}(h_{j\tau}\mathbf{X}_{\cdot\nu}^{\tp}\mathbf{G}\mathbf{X}_{\cdot k}+\mathbf{H}_{j\cdot}\mathbf{X}_{\cdot\nu}\mathbf{G}_{\tau\cdot}\mathbf{X}_{\cdot k}+\mathbf{H}_{j\cdot}\mathbf{X}\mathbf{X}^{\tp}\mathbf{G}_{\cdot\tau}),
\end{align*}
and for $\nu<\tau$, 
\begin{align*}
\frac{\partial\Xi_{jk}}{\partial s_{\tau\nu}} &  =\sum_{a}\sum_{b}h_{j\tau}x_{a\nu}g_{ab}x_{bk}+\sum_{\omega}\sum_{b}h_{j\omega}x_{\omega\nu}g_{\tau b}x_{bk}+\sum_{\omega}\sum_{\ell}\sum_{a}h_{j\omega}x_{\omega\ell}x_{a\ell}g_{a\tau}\\
& =h_{j\tau}\mathbf{X}_{\cdot\nu}^{\tp}\mathbf{G}\mathbf{X}_{\cdot k}+\mathbf{H}_{j\cdot}\mathbf{X}_{\cdot\nu}\mathbf{G}_{\tau\cdot}\mathbf{X}_{\cdot k}+\mathbf{H}_{j\cdot}\mathbf{X}\mathbf{X}^{\tp}\mathbf{G}_{\cdot\tau}.
\end{align*}

\section*{Appendix C: Technical Proofs}
\paragraph{Notation.} 
Without loss of generality, we assume $\tdomain=[0,1]$.  We use $\boldsymbol{\Phi}_q$ to denote the column vector $(\bs_{1},\ldots,\bs_{q})^{\tp}$,
and $\boldsymbol{\Phi}_q(t)$ to denote its value evaluated at $t$. The $\ell^{2}$
norm of a vector $\vec v$ is denoted by $\|\vec v\|_{2}$. When $\vec v$
is viewed as a linear functional, its operator norm is denoted by
$\|\vec v\|$. Note that $\|\vec v\|=\|\vec v\|_{2}$. For a matrix
$\vec M$, $\|\vec M\|$ denotes its induced operator norm, while
$\|\vec M\|_{\fronorm}$ denotes its Frobenius norm. For a function
$\mu$ defined on $\tdomain$, its $\ltwo$ norm, denoted by $\|\mu\|_{\ltwo}$,
is defined by $\|\mu\|_{\ltwo}=\{\int_{\tdomain}|\mu(t)|^{2}\diffop t\}^{1/2}$.
For a covariance function $\gamma$, we use $\|\gamma\|_{\ltwo}$ to denote 
its $\ltwo$ norm that is defined by $\|\gamma\|_{\ltwo}=\{\int_{\tdomain^{2}}|\gamma(s,t)|^{2}\diffop s\diffop t\}^{1/2}$.

\begin{proof}[Proof of Theorem \ref{thm: mean1}] In the sequel, we use $\vec a$ to denote the vector of the coefficients of $\mu_0$ with respect to the basis functions $\phi_1,\ldots,\phi_p$.

First, we observe that, $\mathfrak{c}_{1}\mathfrak{c}_{3}\delta/2\leq f_{T}(t)\leq\mathfrak{c}_{2}\mathfrak{c}_{4}\delta$
for all $t\in[0,1]$. Let $C_{1}=\min\{\mathfrak{c}_{1},\mathfrak{c}_{3},\mathfrak{c}_{1}\mathfrak{c}_{3}\delta/2\}$, $C_{2}=\max\{\mathfrak{c}_{2},\mathfrak{c}_{4},\mathfrak{c}_{2}\mathfrak{c}_{4}\delta\}$ and 
 $\varrho_{n}\asymp n^{-1/2}q^{\alpha+1/2}+\tau_q$.
Define 
\[
Q(\vec a)=\frac{1}{nm}\sum_{i=1}^{n}\sum_{j=1}^{m}\{Y_{ij}-\vec a^{\tp}\boldsymbol{\Phi}_q(T_{ij})\}^{2}+\rho H(\vec a^{\tp}\boldsymbol{\Phi}_q).
\]
Then we observe that
\begin{align}
 & Q(\vec a+\varrho_{n}\vec u)-Q(\vec a)\nonumber \\
= & -\frac{2\varrho_{n}}{nm}\sum_{i=1}^{n}\sum_{j=1}^{m}\vec u^{\tp}\boldsymbol{\Phi}_q(T_{ij})\{Y_{ij}-\vec a^{\tp}\boldsymbol{\Phi}_q(T_{ij})\}+\frac{\varrho_{n}^{2}}{nm}\sum_{i=1}^{n}\sum_{j=1}^{m}\{\vec u^{\tp}\boldsymbol{\Phi}_q(T_{ij})\}^{2}\nonumber \\
 & +2\rho\varrho_{n}\vec u^{\tp}\matvec W\vec a+\rho\varrho_{n}^{2}\vec u^{\tp}\matvec W\vec u\nonumber \\
\equiv & -2\varrho_{n}\mathrm{I}+\varrho_{n}^{2}\mathrm{II}+2\rho\varrho_{n}\mathrm{III}+\rho\varrho_{n}^{2}\mathrm{IV}.\label{eq:mean-component}
\end{align}
It is easy to check that $\|\vec W\|\leq\|\vec W\|_{\mathrm{F}}=O(q^{4\beta+1})$.
Thus, $|\mathrm{III}|=O(q^{4\beta+1})\|\vec u\|_{2}$ and $|\mathrm{IV}|=O(q^{4\beta+1})\|\vec u\|_{2}^{2}$.
According to Claim \ref{claim:mean-I} and \ref{claim:mean-II}, we have
that, for any $\epsilon>0$, there exists $N_{\epsilon}>0$, $\theta_{\epsilon}>0$,
and $\Omega_{n,\epsilon}\subset\Omega$, such that for all $n\geq N_{\epsilon}$, 
\begin{itemize}
\item $\prob(\Omega_{n,\epsilon})\geq1-\epsilon$, and 
\item for all $\omega\in\Omega_{n,\epsilon}$, $|\mathrm{I}|=|\mathrm{I}(\omega)|\leq\absconst_{\epsilon}(n^{-1/2}q^{\alpha+1/2+\tau_q})\|\vec u\|_{2}$,
and 
\item for all $\omega\in\Omega_{n,\epsilon}$, $|\mathrm{II}|=|\mathrm{II}(\omega)|\geq C_{1}\|\vec u\|_{2}^{2}/2$,
\end{itemize}
where $\Omega$ denotes the sample space. With the choice of $\varrho_{n}$ and $\rho$, we can see that $Q(\vec a+\varrho_{n}\vec u)-Q(\vec a)>0$
on $\Omega_{n,\epsilon}$ for all $\vec u$ with $\|\vec u\|_{2}=D_\epsilon$
for a sufficiently large and fixed constant $D_\epsilon>0$. Therefore, with
probability tending to one, the minimizer $\hat{\vec{a}}$ of $Q$ falls into the ball $\{\vec u\in\mathbb{R}^{q}:\|\vec u-\vec a\|_{2}\leq D_\epsilon\varrho_{n}\}$. Now, 
\[
\|\hat{\mu}-\mu_{0}\|_{\ltwo}^{2}\leq2\|\hat{\mu}-\vec a^{\tp}\boldsymbol{\Phi}_q\|_{\ltwo}^{2}+2\|\vec a^{\tp}\boldsymbol{\Phi}_q-\mu_{0}\|_{\ltwo}^{2}=\Op(\varrho_{n}^{2}+\tau_{q}^{2}),
\]
and the proof of the theorem is completed.
\end{proof}

\begin{claim}
\label{claim:mean-I}For any $\epsilon>0$, there exists $N_{\epsilon}>0$,
$\theta_{\epsilon}>0$, and $\Omega_{n,\epsilon}\subset\Omega$, such
that for all $n\geq N_{\epsilon}$, 
\end{claim}

\begin{itemize}
\item $\prob(\Omega_{n,\epsilon})\geq1-\epsilon$, and 
\item for all $\omega\in\Omega_{n,\epsilon}$, $|\mathrm{I}|=|\mathrm{I}(\omega)|\leq\absconst_{\epsilon}(n^{-1/2}q^{\alpha+1/2}+\tau_q)\|\vec u\|_{2}$,
where $\mathrm{I}$ is given in (\ref{eq:mean-component}).
\end{itemize}
\begin{proof}
We first observe that
\begin{align*}
\mathrm{I} & =\frac{1}{nm}\sum_{i=1}^{n}\sum_{j=1}^{m}\vec u^{\tp}\boldsymbol{\Phi}_q(T_{ij})\{Y_{ij}-\mu_{0}(T_{ij})+\mu_{0}(T_{ij})-\vec a^{\tp}\boldsymbol{\Phi}_q(T_{ij})\}\\
 & =\frac{1}{nm}\sum_{i=1}^{n}\sum_{j=1}^{m}\vec u^{\tp}\boldsymbol{\Phi}_q(T_{ij})\{Y_{ij}-\mu_{0}(T_{ij})\}+\frac{1}{nm}\sum_{i=1}^{n}\sum_{j=1}^{m}\vec u^{\tp}\boldsymbol{\Phi}_q(T_{ij})\{\mu_{0}(T_{ij})-\vec a^{\tp}\boldsymbol{\Phi}_q(T_{ij})\}\\
 & \equiv\mathrm{I}_{1}+\mathrm{I}_{2}.
\end{align*}

Now we consider the first term $\mathrm{I}_{1}$. Let $\vec s=\frac{1}{nm}\sum_{i=1}^{n}\sum_{j=1}^{m}\{Y_{ij}-\mu_{0}(T_{ij})\}\boldsymbol{\Phi}_q(T_{ij})$.
We treat it as a linear functional on the Euclidean space $(\real^{q},\|\cdot\|_2)$. It is seen that $\expect\vec s=0$.
Note that the operator norm of $\vec s$ is equal to its $\ell^{2}$ norm. Then, 
\begin{align*}
\var(\|\vec s\|) &\leq \expect\|\vec s\|^2 =\expect\|\vec s\|_{2}^{2}\\
 & =\sum_{1\leq l\leq q}\expect\left(\frac{1}{nm}\sum_{i=1}^{n}\sum_{j=1}^{m}\{Y_{ij}-\mu_{0}(T_{ij})\}\bs_{l}(T_{ij})\right)^{2}\\
 & =\frac{1}{n}\sum_{1\leq l\leq q}\expect\left(\frac{1}{m}\sum_{j=1}^{m}\{Y_{1j}-\mu_{0}(T_{1j})\}\bs_{l}(T_{1j})\right)^{2}\\
 & =O(n^{-1}q^{2\alpha+1}),
\end{align*}
where the last equation is obtained with the aid of the assumption that $\|\phi_k\|_\infty=O(k^\alpha)$ uniformly over all $k$, as the basis is assumed to be an $(\alpha,\beta)$-basis. 
The above derivation also shows that $\expect\|\vec s\|_{2}\leq(\expect\|\vec s\|_{2}^{2})^{1/2}=O(n^{-1/2}q^{\alpha+1/2})$,
Thus, 
\begin{equation}
\sup_{\vec u\neq 0}\frac{|\mathrm{I}_{1}|}{\|\vec u\|_{2}}= \sup_{\vec u\neq 0}\frac{\|\vec s^\tp\vec u\|_2}{\|\vec u\|_{2}}\leq\|\vec s\|_{2}=\Op(n^{-1/2}q^{\alpha+1/2}).\label{eq:pf-mu-1}
\end{equation}

Similar argument shows that 
\begin{equation}\label{eq:pf-mu-11} 
\mathrm{I}_2-\expect\mathrm{I}_2=\Op(n^{-1/2}q^{\alpha+1/2}).
\end{equation}
In addition, 
\begin{align*}
|\expect\mathrm{I}_2| & = |\expect \vec u^{\tp}\boldsymbol{\Phi}_q(T_{11})\{\mu_{0}(T_{11})-\vec a^{\tp}\boldsymbol{\Phi}_q(T_{11})\}|\\
& \leq [\expect\{\vec u^{\tp}\boldsymbol{\Phi}_q(T_{11})\}^2]^{1/2} [\expect\{\mu_{0}(T_{11})-\vec a^{\tp}\boldsymbol{\Phi}_q(T_{11})\}^2]^{1/2}\\
& =O(1)\|\vec u\|_2 \cdot O(1)\|\mu_0-\vec a^{\tp}\boldsymbol{\Phi}_q\|_{\ltwo}\\
&= O(\tau_q)\|\vec u\|_2.
\end{align*}
Combined with \eqref{eq:pf-mu-1} and \eqref{eq:pf-mu-11}, this proves the claim.
\end{proof}

\begin{claim}
\label{claim:mean-II}Suppose $q^{2\alpha+2}n^{-1}\rightarrow0$.
Then, for any $\epsilon>0$, there exists $N_{\epsilon}>0$, 
and $\Omega_{n,\epsilon}\subset\Omega$, such that for all $n\geq N_{\epsilon}$, 
\end{claim}

\begin{itemize}
\item $\prob(\Omega_{n,\epsilon})\geq1-\epsilon$, and 
\item for all $\omega\in\Omega_{n,\epsilon}$, $|\mathrm{II}|=|\mathrm{II}(\omega)|\geq C_{1}\|\vec u\|_{2}^{2}/2$,
where $\mathrm{II}$ is given in (\ref{eq:mean-component}).
\end{itemize}
\begin{proof}
Let $\matvec B=n^{-1}m^{-1}\sum_{i}\sum_{j}\boldsymbol{\Phi}_q^{\tp}(T_{ij})\boldsymbol{\Phi}_q(T_{ij})$
and $\matvec E=\expect\matvec B$. Now we derive the bound for the
smallest eigenvalue $\underline{\Lambda}(\matvec E)$ of $\matvec E$. As $\matvec E$
is positive semi-positive definite, one has 
\[
\underline{\Lambda}(\matvec E)=\inf_{\|\vec u\|_{2}=1}\vec u^{\tp}\matvec E\vec u.
\]
Let $E_{jk}$ be the element of $\matvec E$ at the $j$th row and
$k$th column. Note that $$E_{jk}=\expect\{\bs_{j}(T)\bs_{k}(T)\}=\int_{0}^{1}\bs_{j}(t)\bs_{k}(t)f_T(t)\diffop t.$$
Thus,
\begin{align*}
\vec u^{\tp}\matvec E\vec u & =\sum_{j,k}u_{j}E_{jk}u_{k} =\int_{0}^{1}f_T(t)\left(\sum_{j}u_{j}\bs_{j}(t)\right)^{2}\diffop t\\
 & \geq C_{1}\int_{0}^{1}\left(\sum_{j}u_{j}\bs_{j}(t)\right)^{2}\diffop t=C_{1}\|\vec u\|_{2}^{2},
\end{align*}
where the last inequality is due to the orthonormality of $\{\bs_{k}\}$.
Thus, 
\[
\underline{\Lambda}(\matvec E)\geq C_{1}.
\]

Observe that $\matvec B=\matvec E+\Delta$ with $\Delta=\matvec B-\matvec E$.
We then have $\expect\Delta=\vec 0$ and 
\begin{align*}
\var(\|\Delta\|) & \leq\expect\|\Delta\|^{2}\leq\expect\|\Delta\|_{\fronorm}^{2}\\
 & \leq\sum_{1\leq l,h\leq q}\expect\left(n^{-1}m^{-1}\sum_{i}\sum_{j}\bs_{l}(T_{ij})\bs_{h}(T_{ij})-E_{lh}\right)^{2}\\
 & =\sum_{1\leq l,h\leq q}\frac{1}{n}\expect\left(\frac{1}{m}\sum_{j}\bs_{l}(T_{1j})\bs_{h}(T_{1j})-E_{lh}\right)^{2}\\
 & =O(q^{2\alpha+2}n^{-1}).
\end{align*}
Now the claim follows from Weyl's inequality.
\end{proof}

\begin{proof}[Proof of Theorem \ref{thm:2} and \ref{thm:3}] Theorem \ref{thm:2} is a special case of Theorem \ref{thm:3} with $\DD=0$. Below we prove Theorem \ref{thm:3}. In the sequel, we use  $\tilde{\vec C}$ to denote the matrix formed by the  coefficients of $\tilde\gamma$ with respect to the basis functions $\{\phi_k\otimes\phi_l:\,1\leq k,l\leq p\}$.

Let $\varrho_{n}=p^{2\alpha+1}n^{-1/2}+\kappa_{p}+\DD$ and 
\begin{equation}
\mathcal{B}_{p}(\theta)=\{\matvec D:\|\matvec D\|_{\fronorm}=\theta,\tilde{\matvec C}+\matvec D\in\mathfrak{C}(p)\},\label{eq:B-ball}
\end{equation}
where ${\mathfrak{C}}(p)$ denotes the collection of $p\times p$ symmetric matrices $\vec G=[g_{kl}]$ such that $\gamma_{\vec G}=\sum_{1\leq k,l\leq p}g_{kl}\phi_k\otimes\phi_l\in\mathcal{C}$. For any $p\times p$ matrix $\vec C$, define 
\[
Q(\matvec C)=\frac{1}{nm(m-1)}\sum_{i=1}^{n}\sum_{1\leq j\neq k\leq m}\{\Gamma_{ijk}-\gamma_{\matvec C}(T_{ij},T_{ik})\}^{2}+\lambda J(\gamma_{\matvec C}).
\]
Then we observe that 
\begin{align}
 & Q(\tilde{\matvec C}+\varrho_{n}\matvec D)-Q(\tilde{\matvec C})\nonumber \\
= & -2\varrho_{n}\frac{1}{nm(m-1)}\sum_{i=1}^{n}\sum_{1\leq j\neq k\leq m}\{\Gamma_{ijk}-\gamma_{\tilde{\vec C}}(T_{ij},T_{ik})\}\boldsymbol{\Phi}_p^{\tp}(T_{ij})\vec D\boldsymbol{\Phi}_{p}(T_{ik})\nonumber \\
 & +\varrho_{n}^{2}\frac{1}{nm(m-1)}\sum_{i=1}^{n}\sum_{1\leq j\neq k\leq m}\{\boldsymbol{\Phi}_{p}^{\tp}(T_{ij})\vec D\boldsymbol{\Phi}_{p}(T_{ik})\}^{2}\nonumber \\
 & +2\lambda\varrho_{n}\{\tr(\vec{DW}\tilde{\vec C}\vec U)+\tr(\vec D\vec V\tilde{\vec C}\vec V)\}\nonumber \\
 & +\lambda\varrho_{n}^{2}\{\tr(\vec{DU}\vec D\vec W)+\tr(\vec D\vec V\vec D\vec V)\}\nonumber \\
\equiv & -2\varrho_{n}\mathrm{I}+\varrho_{n}^{2}\mathrm{II}+2\lambda\varrho_{n}\mathrm{III}+\lambda\varrho_{n}^{2}\mathrm{IV}.\label{eq:cov-component}
\end{align}
Note that  $\tr(\vec A)\leq \sqrt p \|\vec A\|_{\fronorm}$ for any $p\times p$ matrix $\vec A$. With the the assumption that $\|\phi_k^{(r)}\|_{\ltwo}=O(k^{\beta r})$ for $r=1,2$, one can then check that $|\tr(\vec{DW}\tilde{\vec C}\vec U)|=O(p^{4\beta+5/2})\|\vec D\|_{\fronorm}$,
$|\tr(\vec D\vec V\tilde{\vec C}\vec V)|=O(p^{4\beta+5/2})\|\vec D\|_{\fronorm}$,
$|\tr(\vec{DU}\vec D\vec W)|=O(p^{4\beta+5/2})\|\vec D\|_{\fronorm}^{2}$
and $|\tr(\vec D\vec V\vec D\vec V)|=O(p^{4\beta+5/2})\|\vec D\|_{\fronorm}^{2}$
by using the fact that $\|\vec U\|_{\fronorm}^{2}=O(p^2)$, $\|\vec W\|_{\fronorm}^{2}=O(p^{8\beta+2})$, $\|\vec V\|_{\fronorm}^{2}=O(p^{4\beta+2})$ and $\|\tilde{\vec C}\|_{\fronorm}\leq \|\gamma_{0}\|_{L^2}+\DD=O(1)$. These give the orders of the terms $\mathrm{III}$ and $\mathrm{IV}$, which are $O(p^{4\beta+5/2})\|\vec D\|_{\fronorm}$ and $O(p^{4\beta+5/2})\|\vec D\|_{\fronorm}^2$, respectively. 

According to Claim \ref{claim:cov-I}
and \ref{claim:cov-II}, now we have that, for any $\epsilon>0$,
there exists $N_{\epsilon}>0$, $\theta_\epsilon>0$, and $\Omega_{n,\epsilon}\subset\Omega$,
such that for all $n\geq N_{\epsilon}$, 
\begin{itemize}
\item $\prob(\Omega_{n,\epsilon})\geq1-\epsilon$, and 
\item for all $\omega\in\Omega_{n,\epsilon}$, $|\mathrm{I}|=|\mathrm{I}(\omega)|\leq\absconst_{\epsilon}(p^{2\alpha+1}n^{-1/2}+\kappa_p+\DD)\|\vec D\|_{\fronorm}$,
and 
\item for all $\omega\in\Omega_{n,\epsilon}$, $|\mathrm{II}|=|\mathrm{II}(\omega)|\geq\absconst_{1}\|\vec D\|_{\fronorm}^{2}$, where the constant $\theta_1>0$ depends only on $f_R$, $f_{T\mid R}$, $\gamma_0$ and $\mathcal{C}$.
\end{itemize}
With the choice of $\varrho_{n}$ and $\lambda$, we can see that
$Q(\tilde{\vec C}+\varrho_{n}\vec D)-Q(\tilde{\vec C})>0$ on $\Omega_{n,\epsilon}$
for all $\vec D\in\mathcal{B}_{p}(\theta)$ for all  sufficiently large
and fixed constant $\absconst>0$. Therefore, with probability tending
to one, the minimizer $\hat{\vec C}$ of $Q$ falls into  $\{\vec G\in\mathfrak{C}(p):\|\vec G-\tilde{\vec C}\|_{\fronorm}\leq\varrho_{n}\absconst\}$.
Therefore, 
\[
\|\hat{\gamma}-\gamma_{0}\|_{\hsnorm}\leq\|\hat{\gamma}-\boldsymbol{\Phi}_{p}^{\tp}\tilde{\matvec C}\boldsymbol{\Phi}_{p}\|_{\ltwo}+\|\boldsymbol{\Phi}_{p}^{\tp}\tilde{\matvec C}\boldsymbol{\Phi}_{p}-\tilde\gamma\|_{\hsnorm}+\|\tilde{\gamma}-\gamma_{0}\|_{\hsnorm}=\Op(\varrho_{n}+\kappa_{p}+\DD),
\]
and the proof of the theorem is completed.
\end{proof}

\begin{claim}
\label{claim:cov-I}For any $\epsilon>0$, there exists $N_{\epsilon}>0$,
$\theta_{\epsilon}>0$, and $\Omega_{n,\epsilon}\subset\Omega$, such
that for all $n\geq N_{\epsilon}$, 
\end{claim}

\begin{itemize}
\item $\prob(\Omega_{n,\epsilon})\geq1-\epsilon$, and 
\item for all $\omega\in\Omega_{n,\epsilon}$, $|\mathrm{I}|=|\mathrm{I}(\omega)|\leq\absconst_{\epsilon}(p^{2\alpha+1}n^{-1/2}+\kappa_{p}+\DD)\|\vec D\|_{\fronorm}$,
where $\mathrm{I}$ is given in (\ref{eq:cov-component}).
\end{itemize}
\begin{proof}
We first observe that
\begin{align*}
\mathrm{I}= & \frac{1}{nm(m-1)}\sum_{i=1}^{n}\sum_{1\leq j\neq k\leq m}\{\Gamma_{ijk}-\gamma_{0}(T_{ij},T_{ik})\}\boldsymbol{\Phi}_{p}^{\tp}(T_{ij})\vec D\boldsymbol{\Phi}_{p}(T_{ik})\\
 & +\frac{1}{nm(m-1)}\sum_{i=1}^{n}\sum_{1\leq j\neq k\leq m}\{\gamma_{0}(T_{ij},T_{ik})-\gamma_{\tilde{\vec C}}(T_{ij},T_{ik})\}\boldsymbol{\Phi}_{p}^{\tp}(T_{ij})\vec D\boldsymbol{\Phi}_{p}(T_{ik})\\
 \equiv & \mathrm{I}_{1}+\mathrm{I}_{2}.
\end{align*}
Now we consider the first term $\mathrm{I}_{1}$. Let $\vec S=\frac{1}{nm(m-1)}\sum_{i=1}^{n}\sum_{1\leq j\neq k\leq m}[\{Y_{ij}-\mu_{0}(T_{ij})\}\{Y_{ik}-\mu_{0}(T_{ik})\}-\gamma_{0}(T_{ij},T_{ik})]\boldsymbol{\Phi}_p^\tp (T_{ij})\odot \boldsymbol{\Phi}_p^\tp (T_{ik})$, where $\odot$ denotes the Kronecker product of matrices.
This $\vec S$ is viewed as a random linear functional acting on $\overrightarrow{\vec D}\in\real^{p^2}$, where $\overrightarrow{\vec D}$ denotes 
the vectorization of $\vec D$ obtained by stacking the columns of $\vec D$ into a single column vector.  To quantify the order of its operator norm $\|\vec S\|$, we first observe that $\expect \vec S=0$. Noting the operator norm is bounded by the Frobenius norm, we  deduce that 
\begin{align*}
& \var(\|\vec S\|) \leq\expect\|\vec S\|_{\fronorm}^{2}\\
 & =\frac{1}{n}\sum_{1\leq l,h\leq p}\expect\left(\frac{1}{m(m-1)}\sum_{1\leq j\neq k\leq m}[\{Y_{ij}-\mu_{0}(T_{ij})\}\{Y_{ik}-\mu_{0}(T_{ik})\}-\gamma_{0}(T_{ij},T_{ik})]\bs_{l}(T_{ij})\bs_{h}(T_{ik})\right)^{2}\\
 & \leq\frac{1}{n}\frac{1}{m(m-1)}\sum_{1\leq l,h\leq p}\sum_{1\leq j\neq k\leq m}\expect\left([\{Y_{1j}-\mu_{0}(T_{1j})\}\{Y_{1k}-\mu_{0}(T_{1k})\}-\gamma_{0}(T_{1j},T_{1k})]\bs_{l}(T_{1j})\bs_{h}(T_{1k})\right)^{2}\\
 & =O(n^{-1})\sum_{1\leq l,h\leq p}l^{2\alpha}h^{2\alpha}\sum_{1\leq j\neq k\leq m}\expect\left([\{Y_{1j}-\mu_{0}(T_{1j})\}\{Y_{1k}-\mu_{0}(T_{1k})\}-\gamma_{0}(T_{1j},T_{1k})]\right)^{2}\\
 & =O\left(\frac{p^{4\alpha+2}}{n}\right)\expect[\{Y_{11}-\mu_{0}(T_{11})\}\{Y_{12}-\mu_{0}(T_{12})\}-\gamma_{0}(T_{11},T_{12})]^{2}\\
 & =O\left(\frac{p^{4\alpha+2}}{n}\right)(\expect[\{Y_{11}-\mu_{0}(T_{11})\}\{Y_{12}-\mu_{0}(T_{12})\}]^2+\expect[\gamma_{0}(T_{11},T_{12})]^{2})\\
 & =O(p^{4\alpha+2}n^{-1}),
\end{align*}
where the last equation is due to the fact that $\gamma_0$ is continuous and thus bounded on $[0,1]^2$, and the fact that, with the notation $Z_1=X_1-\mu_0$, 
\begin{align*}
& \expect[\{Y_{11}-\mu_{0}(T_{11})\}\{Y_{12}-\mu_{0}(T_{12})\}]^2\\
& = \expect[\{Z_{1}(T_{11})+\varepsilon_{11})\}\{Z_{1}(T_{12})+\varepsilon_{12}\}]^2\\
& \leq 4\expect\{Z_1^2(T_{11})+\varepsilon_{11}^2\}\{Z_1^2(T_{12})+\varepsilon_{12}^2\}\\
& =4\expect\{Z_1^2(T_{11})Z_1^2(T_{12})\}+O(\expect\|Z_1\|^2_{\ltwo})+O(1)\\
& =4\expect[\expect\{Z_1^2(T_{11})\mid O_1,Z_1\}\expect\{Z_1^2(T_{12})\mid O_1,Z_1\}]+O(\expect\|X\|^2_{\ltwo})+O(1)\\
& =O(\expect\{\|Z_1\|^2_{\ltwo} \|Z_1\|^2_{\ltwo}\})+O(\expect\|X\|^2_{\ltwo})+O(1)\\
& =O(\expect\{\|X\|^4_{\ltwo})+O(\expect\|X\|^2_{\ltwo})+O(1)\\
& =O(1),
\end{align*}
since $\expect\|X\|^{4}_{\ltwo}<\infty$ which also implies $\expect \|X\|^2_{\ltwo}<\infty$. Thus, 
\begin{align*}
\sup_{\vec D\neq 0}\frac{|\mathrm{I}_{1}|}{\|\vec D\|_{\fronorm}} = \sup_{\vec D\neq 0}\frac{|\vec S \overrightarrow{\vec D}|}{\|\vec D\|_{\fronorm}} & \leq\|\vec S\|_{\fronorm}=O(p^{2\alpha+1}n^{-1/2}).
\end{align*}
Therefore,
\begin{equation}
\mathrm{I}_{1}=\Op(p^{2\alpha+1}n^{-1/2})\|\vec D\|_{\fronorm},\label{eq:cov-pf-1}
\end{equation}
where the term $\Op(\cdot)$ is uniform over all $\vec D$.

Similar derivation combined with the strategy for establishing \eqref{eq:pf-mu-11} shows that \begin{equation}\label{eq:cov-pf-11}\mathrm{I}_{2}-\expect \mathrm{I}_2=\Op(p^{2\alpha+1}n^{-1/2})\|\vec D\|_{\fronorm}.\end{equation} In addition,
\begin{align*}
|\expect \mathrm{I}_2| &= |\expect[\{\gamma_{0}(T_{11},T_{12})-\gamma_{\tilde{\vec C}}(T_{11},T_{12})\}\boldsymbol{\Phi}_{p}^{\tp}(T_{11})\vec D\boldsymbol{\Phi}_{p}(T_{12})]|\\
& \leq [\expect\{\gamma_{0}(T_{11},T_{12})-\gamma_{\tilde{\vec C}}(T_{11},T_{12})\}^2]^{1/2}[\expect\{\boldsymbol{\Phi}_{p}^{\tp}(T_{11})\vec D\boldsymbol{\Phi}_{p}(T_{12})\}^2]^{1/2}\\
& \leq O(1)\|\gamma_{0}-\gamma_{\tilde{\vec C}}\|_{\fronorm}\cdot O(1)\|\vec D\|_{\fronorm}\\
& \leq O(\kappa_{p}+\DD)\|\vec D\|_{\fronorm}.
\end{align*}
The claim then follows from the above derivation, \eqref{eq:cov-pf-1}, and \eqref{eq:cov-pf-11}.
\end{proof}

\begin{claim}
\label{claim:cov-II}Suppose $p^{4\alpha+2}n^{-1/2}\rightarrow0$.
Then, for any $\epsilon>0$, there exists $N_{\epsilon}>0$ 
and $\Omega_{n,\epsilon}\subset\Omega$, such that for all $n\geq N_{\epsilon}$, 
\end{claim}

\begin{itemize}
\item $\prob(\Omega_{n,\epsilon})\geq1-\epsilon$, and 
\item  for some constant $\absconst_{1}$ depending on $f_{R}$,
$f_{T\mid R}$, $\gamma_{0}$ and $\mathcal{C}$,  for all $\omega\in\Omega_{n,\epsilon}$, $|\mathrm{II}|=|\mathrm{II}(\omega)|\geq\absconst_{1}\|\vec D\|_{\fronorm}^{2}$
for some constant $\absconst_{1}>0$, where $\mathrm{II}$ is given
in (\ref{eq:cov-component}).
\end{itemize}
\begin{proof}
We first establish that $\expect\{\boldsymbol{\Phi}_{p}^{\tp}(T_{11})\vec D\boldsymbol{\Phi}_{p}(T_{12})\}^{2}\geq\absconst_{1}\|\vec D\|_{\fronorm}^{2}$
for some constant $\absconst_{1}$ independent of $\vec D$. Suppose
that this is false. Then there is a sequence
$\xi_{r}\rightarrow0$ and a sequence $\vec D_r\in\mathcal{B}_p(\absconst)$, such that  
\[
\expect\{\boldsymbol{\Phi}_{p}^{\tp}(T_{11})\vec D_{r}\boldsymbol{\Phi}_{p}(T_{12})\}^{2}\leq\xi_{r}\|\vec D_{r}\|_{\fronorm}^{2}=\xi_{r}\absconst^{2}\rightarrow0,
\]
where $\mathcal{B}_p(\theta)$ is defined in (\ref{eq:B-ball}). By Fatou's
lemma, we have
\[
\expect\underset{r\rightarrow\infty}{\underline{\lim}}\{\boldsymbol{\Phi}_{p}^{\tp}(T_{11})\vec D_{r}\boldsymbol{\Phi}_{p}(T_{12})\}^{2}\leq\underset{r\rightarrow\infty}{\underline{\lim}}\expect\{\boldsymbol{\Phi}_{p}^{\tp}(T_{11})\vec D_{r}\boldsymbol{\Phi}_{p}(T_{12})\}^{2}\rightarrow0.
\]
This implies that there exists a subsequence $\{r_{\ell}\}_{\ell=1}^{\infty}$
such that 
\[
\lim_{\ell\rightarrow\infty}\{\boldsymbol{\Phi}_{p}^{\tp}(s)\vec D_{r_{\ell}}\boldsymbol{\Phi}_{p}(t)\}=0\qquad a.e.\,\,\text{on}\,\,\tdomain_{\delta},
\]
or equivalently, 
\begin{equation}\label{eq:pf-T-delta-id}
\lim_{\ell\rightarrow\infty}\{\boldsymbol{\Phi}_{p}^{\tp}(s)(\tilde{\vec C} +\vec D_{r_{\ell}})\boldsymbol{\Phi}_{p}(t)\}=\boldsymbol{\Phi}_{p}^{\tp}(s)\tilde{\vec C}\boldsymbol{\Phi}_{p}(t)=\gamma_{\tilde{\vec C}}(s,t)\qquad a.e.\,\,\text{on}\,\,\tdomain_{\delta},
\end{equation}
where $\tdomain_{\delta}=\{(s,t):s,t\in[0,1],|s-t|<\delta\}$. Furthermore,
the uniform boundedness of the sequence $\boldsymbol{\Phi}_{p}^{\tp}(s)(\tilde{\vec C} +\vec D_{r_{\ell}})\boldsymbol{\Phi}_{p}(t)$
implies that there exists a further subsequence $r_{\ell_{h}}$ such
that $\boldsymbol{\Phi}_{p}^{\tp}(s)(\tilde{\vec C} +\vec D_{r_{\ell_{h}}})\boldsymbol{\Phi}_{p}(t)$ converges
pointwisely to some $\psi\in\mathcal{C}$, since $\boldsymbol{\Phi}_{p}^{\tp}(\tilde{\vec C}+\vec D_{r_{\ell_{h}}})\boldsymbol{\Phi}_{p}\in\mathcal{C}$ due to $\vec D_{r_{\ell_h}}\in\mathcal{B}_p(\theta)$ 
and we recall that $\mathcal{C}$ is a BSC family. This $\psi$ ought to be $\psi(s,t)=\gamma_{\tilde{\vec C}}(s,t)$ for all $s,t\in[0,1]^{2}$, since $\mathcal{C}$ is a $\tdomain_\delta$-identifiable family and $\psi$  agrees with $\gamma_{\tilde{\vec C}}$ on $\tdomain_\delta$ according to \eqref{eq:pf-T-delta-id}. Now 
Fatou's lemma suggests that 
\begin{align*}
\absconst^{2} & =\|\vec D_{r_{\ell_{h}}}\|_{\fronorm}^{2}=\underset{h\rightarrow\infty}{\overline{\lim}}\int_{0}^{1}\int_{0}^{1}\{\boldsymbol{\Phi}_{p}^{\tp}(s)\vec D_{r_{\ell_{h}}}\boldsymbol{\Phi}_{p}(t)\}^{2}\diffop s\diffop t\\
 & \leq\int_{0}^{1}\int_{0}^{1}\underset{h\rightarrow\infty}{\overline{\lim}}\{\boldsymbol{\Phi}_{p}^{\tp}(s)\vec D_{r_{\ell_{h}}}\boldsymbol{\Phi}_{p}(t)\}^{2}\diffop s\diffop t\\
 & = \int_{0}^{1}\int_{0}^{1}\underset{h\rightarrow\infty}{\overline{\lim}}\{\boldsymbol{\Phi}_{p}^{\tp}(s)(\tilde{\vec C}+\vec D_{r_{\ell_{h}}})\boldsymbol{\Phi}_{p}(t)-\boldsymbol{\Phi}_{p}^{\tp}(s)\tilde{\vec C}\boldsymbol{\Phi}_{p}(t)\}^{2}\diffop s\diffop t\\
 & = \int_0^1\int_0^1\{\psi(s,t)-\gamma_{\tilde{\vec C}}(s,t)\}^2\diffop s\diffop t\\
 & =0,
\end{align*}
which contradicts with $\absconst>0$.

Now we write $\vec r$ for the column vector of $\{\bs_{k}\otimes\phi_{l}:1\leq k,l\leq p\}$
and define $\mathbf{R}=n^{-1}m^{-1}(m-1)^{-1}\sum_{i=1}^{n}\sum_{1\leq j\neq k\leq m}\vec r(T_{ij},T_{ik})\vec r^{\tp}(T_{ij},T_{ik})$. This $\vec R$ is viewed as a linear operator on $(\real^{p^2},\|\cdot\|_2)$. The above result shows that  $\expect(\overrightarrow{\vec D}^{\tp}\vec R\overrightarrow{\vec D})\geq \theta_1\|\vec D\|_{\fronorm}^2$, 
where $\overrightarrow{\vec D}$ denotes the vectorization of $\vec D$. More precisely,  
\begin{equation}\label{eq:pf-claim-4-1}
\inf_{\vec D\in\mathcal{B}_p(\theta)}\frac{\overrightarrow{\vec D}^{\tp}(\expect\vec R)\overrightarrow{\vec D}}{\|\vec D\|_{\fronorm}^{2}}\geq\theta_{1}.
\end{equation}

Let $\Delta=\vec R-\expect\vec R$. We observe that 
\begin{equation}\label{eq:pf-claim-4-2}
\sup_{\vec D\in\mathcal{B}_p(\theta)}\frac{\overrightarrow{\vec D}^{\tp}\Delta\overrightarrow{\vec D}}{\|\vec D\|_{\fronorm}^{2}}\leq\sup_{\vec D}\frac{\overrightarrow{\vec D}^{\tp}\Delta\overrightarrow{\vec D}}{\|\vec D\|_{\fronorm}^{2}}\leq\|\Delta\|\leq\|\Delta\|_{\fronorm}=\Op(p^{4\alpha+2}n^{-1/2}),
\end{equation}
since
\begin{align*}
\expect\|\Delta\|_{\fronorm}^{2} & \leq\sum_{1\leq l,h,r,s\leq p}\expect\left(\frac{1}{nm(m-1)}\sum_{i=1}^{n}\sum_{1\leq j\neq k\leq m}\bs_{l}(T_{ij})\bs_{h}(T_{ik})\bs_{r}(T_{ij})\bs_{s}(T_{ik})-R_{lhrs}\right)^{2}\\
 & \leq\frac{1}{n}\sum_{1\leq l,h,r,s\leq p}\expect\left(\frac{1}{m(m-1)}\sum_{1\leq j\neq k\leq m}\bs_{l}(T_{1j})\bs_{h}(T_{1k})\bs_{r}(T_{1j})\bs_{s}(T_{1k})-R_{lhrs}\right)^{2}\\
 & =O(p^{8\alpha+4}/n),
\end{align*}
where $R_{lhrs}=\expect\{\bs_{l}(T_{ij})\bs_{h}(T_{ik})\bs_{r}(T_{ij})\bs_{s}(T_{ik})\}=O(p^{4\alpha})$
uniform over all $l$, $h$, $r$ and $s$. Now the conclusion of
the claim follows from \eqref{eq:pf-claim-4-1}, \eqref{eq:pf-claim-4-2}, and the observation $\mathrm{I}_2=\overrightarrow{\vec D}^{\tp}\vec R\overrightarrow{\vec D}=\overrightarrow{\vec D}^{\tp}(\expect\vec R)\overrightarrow{\vec D}+\overrightarrow{\vec D}^{\tp}\Delta\overrightarrow{\vec D}$.
\end{proof}

\bibliographystyle{lin}
\bibliography{snippet}

\end{document}